\PassOptionsToPackage{unicode}{hyperref}
\PassOptionsToPackage{hyphens}{url}
\PassOptionsToPackage{dvipsnames,svgnames,x11names}{xcolor}
\documentclass[
  12pt]{article}

\usepackage{amsmath,amssymb}
\usepackage{iftex}
\ifPDFTeX
  \usepackage[T1]{fontenc}
  \usepackage[utf8]{inputenc}
  \usepackage{textcomp} 
\else 
  \usepackage{unicode-math}
  \defaultfontfeatures{Scale=MatchLowercase}
  \defaultfontfeatures[\rmfamily]{Ligatures=TeX,Scale=1}
\fi
\usepackage{lmodern}
\ifPDFTeX\else  
\fi
\IfFileExists{upquote.sty}{\usepackage{upquote}}{}
\IfFileExists{microtype.sty}{
  \usepackage[]{microtype}
  \UseMicrotypeSet[protrusion]{basicmath} 
}{}
\makeatletter
\@ifundefined{KOMAClassName}{
  \IfFileExists{parskip.sty}{%
    \usepackage{parskip}
  }{
    \setlength{\parindent}{0pt}
    \setlength{\parskip}{6pt plus 2pt minus 1pt}}
}{
  \KOMAoptions{parskip=half}}
\makeatother
\usepackage{xcolor}
\makeatletter
\ifx\paragraph\undefined\else
  \let\oldparagraph\paragraph
  \renewcommand{\paragraph}{
    \@ifstar
      \xxxParagraphStar
      \xxxParagraphNoStar
  }
  \newcommand{\xxxParagraphStar}[1]{\oldparagraph*{#1}\mbox{}}
  \newcommand{\xxxParagraphNoStar}[1]{\oldparagraph{#1}\mbox{}}
\fi
\ifx\subparagraph\undefined\else
  \let\oldsubparagraph\subparagraph
  \renewcommand{\subparagraph}{
    \@ifstar
      \xxxSubParagraphStar
      \xxxSubParagraphNoStar
  }
  \newcommand{\xxxSubParagraphStar}[1]{\oldsubparagraph*{#1}\mbox{}}
  \newcommand{\xxxSubParagraphNoStar}[1]{\oldsubparagraph{#1}\mbox{}}
\fi
\makeatother

\usepackage{longtable,booktabs,array}
\usepackage{calc} 
\usepackage{etoolbox}
\makeatletter
\patchcmd\longtable{\par}{\if@noskipsec\mbox{}\fi\par}{}{}
\makeatother
\IfFileExists{footnotehyper.sty}{\usepackage{footnotehyper}}{\usepackage{footnote}}
\makesavenoteenv{longtable}
\usepackage{graphicx}
\makeatletter
\def\maxwidth{\ifdim\Gin@nat@width>\linewidth\linewidth\else\Gin@nat@width\fi}
\def\maxheight{\ifdim\Gin@nat@height>\textheight\textheight\else\Gin@nat@height\fi}
\makeatother
\setkeys{Gin}{width=\maxwidth,height=\maxheight,keepaspectratio}
\makeatletter
\def\fps@figure{htbp}
\makeatother

\makeatletter
\@ifpackageloaded{caption}{}{\usepackage{caption}}
\AtBeginDocument{%
\ifdefined\contentsname
  \renewcommand*\contentsname{Table of contents}
\else
  \newcommand\contentsname{Table of contents}
\fi
\ifdefined\listfigurename
  \renewcommand*\listfigurename{List of Figures}
\else
  \newcommand\listfigurename{List of Figures}
\fi
\ifdefined\listtablename
  \renewcommand*\listtablename{List of Tables}
\else
  \newcommand\listtablename{List of Tables}
\fi
\ifdefined\figurename
  \renewcommand*\figurename{Figure}
\else
  \newcommand\figurename{Figure}
\fi
\ifdefined\tablename
  \renewcommand*\tablename{Table}
\else
  \newcommand\tablename{Table}
\fi
}

\@ifpackageloaded{float}{}{\usepackage{float}}
\floatstyle{ruled}
\@ifundefined{c@chapter}{\newfloat{codelisting}{h}{lop}}{\newfloat{codelisting}{h}{lop}[chapter]}
\floatname{codelisting}{Listing}

\makeatother
\makeatletter
\@ifpackageloaded{caption}{}{\usepackage{caption}}
\@ifpackageloaded{subcaption}{}{\usepackage{subcaption}}
\makeatother

\ifLuaTeX
  \usepackage{selnolig}  
\fi
\usepackage{bookmark}

\usepackage[autostyle]{csquotes}
\usepackage{arydshln}
\usepackage{amsthm}
\newtheorem{theorem}{Theorem}[section]
\newtheorem{lemma}[theorem]{Lemma}
\newtheorem{proposition}{Proposition}
\usepackage{enumerate} 

\IfFileExists{xurl.sty}{\usepackage{xurl}}{} 
\hypersetup{
  pdftitle={From Unsupervised to Guided Clustering: A Variational Implementation},
  pdfauthor={Violaine Courrier, Christophe Biernacki},
  pdfkeywords={clustering, variational inference},
  colorlinks=true,
  linkcolor={blue},
  filecolor={Maroon},
  citecolor={Blue},
  urlcolor={Blue},
  pdfcreator={LaTeX via pandoc}}

\newcommand{\anon}{1}

\begin{document}

\def\spacingset#1{\renewcommand{\baselinestretch}%
{#1}\small\normalsize} \spacingset{1}


\if1\anon
{
  \title{\bf From Unsupervised to Guided Clustering: \\ A Variational Implementation}
  \author{Violaine Courrier\thanks{Corresponding author: violaine.courrier@inria.fr}\\
    Withings, Inria, Université de Lille, CNRS\\
    and \\
    Christophe Biernacki \\
    Inria, Université de Lille, CNRS}
  \date{}
  \maketitle
} \fi

\if0\anon
{
  \bigskip
  \bigskip
  \bigskip
  \begin{center}
    {\LARGE\bf Title}
\end{center}
  \medskip
} \fi

\bigskip
\begin{abstract}
Clustering is viewed as an unsupervised technique, but in practice it requires guidance to uncover meaningful structures. We formalize this with guided clustering, a paradigm that uses a guiding variable to steer the discovery process, and introduce the Guided Clustering Variational Autoencoder (GCVAE) as its deep generative realization. GCVAE learns a latent space structured as a Gaussian Mixture Model by optimizing a variational objective that forces the representation to be maximally informative about the guiding variable.  This framework allows the resulting clustering to be reoriented by changing the guiding variable, yielding clusters that are meaningful for the specified context. Experiments on public (MNIST-SVHN) and proprietary connected health devices data demonstrate GCVAE's ability to discover coherent and task-relevant clusters in complex settings.\end{abstract}

\noindent%
\vfill
%
%
%
\newpage
\spacingset{1.8} 

\section{Introduction}\label{sec-intro}

Clustering is never a purely unsupervised task. In practice, analysts implicitly guide the discovery process through choices of parameters, distance metrics, or feature selection to find \enquote{interesting} results. This hands-on necessity stems from a more fundamental challenge: the inherent ambiguity of the clustering task itself. A single dataset can contain multiple, equally valid partitions, and the most meaningful one depends entirely on the analytical goal. Without an explicit objective, the search for a relevant partition relies on an often informal trial-and-error guidance. We argue that this external knowledge should not be an afterthought, but a formal component of the model definition.

We propose to formalize this via the so-called \textit{guided clustering}. In this new paradigm, the guiding variable is not the end goal but rather the lens through which we discover the most relevant and coherent partition of the data. We propose a specific implementation of this novel guided clustering approach using a deep generative model that learns to compress the input into a representation that is simultaneously organized into a discrete mixture of clusters and optimized to be maximally predictive of the guiding variable. The Guided Clustering Variational Autoencoder (GCVAE) model directly addresses the limitations of prior works by learning non-linear, structurally meaningful clusters within a compact space tailored by the guiding variable. 

\section{Related works}

The inherent ambiguity of unsupervised clustering often necessitates guidance to discover partitions that are relevant to a specific analytical context. The guiding principle can be broadly categorized based on its source: either internal to the data itself or provided by an external, contextual variable.

Internal guidance seeks to find structure by modifying the input data $\boldsymbol{x}$. The motivation is that not all features are equally relevant. This has led to variable selection techniques in model-based clustering that identify a feature subset to best reveal latent groups \cite{var_selection, maugis_variable_2009,maugis_variable_2011}. More advanced frameworks, like Multi-Partition Clustering, concurrently discover multiple valid partitions, each defined by a different feature subset \cite{marbac2018tractablemultipartitionsclustering}. While powerful, these methods are fundamentally introspective, defining \enquote{interesting} clusters based solely on the statistical properties of the input data $\boldsymbol{x}$.

In contrast, external guidance aligns the clustering process with an extrinsic goal. The necessity of utilizing auxiliary information to resolve the ambiguity of unsupervised representation learning has been formally established in frameworks such as Identifiable VAEs (iVAE) \cite{khemakhem2020variational}. However, while iVAE leverages auxiliary variables to recover continuous latent factors (disentanglement), our focus is on discovering discrete structures. In this context, the most direct form of guidance is semi-supervised clustering, where the discovery is steered by observing a subset of the actual cluster labels \cite{kingmasemisupervised, maaloesemisupervised} or pairwise constraints \cite{basu_constrained_2008}. However, this approach inherently requires partial access to the ground truth partition, which is often unavailable in exploratory settings. When the cluster labels are unobserved, guidance must instead come from a contextual variable.

The most explicit form of this approach, and the focus of this work, is predictive clustering \cite{goc}. This paradigm reframes clustering as a tool for supervised discovery, where the objective is to cluster $\boldsymbol{x}$ into groups that are predictive of a target variable $\boldsymbol{y}$. The pursuit of this objective has led to the development of diverse methods. For instance, the Predictive Clustering Trees (PCTs), which adapt decision tree algorithms to recursively partition the joint data space \cite{clusteringtrees, obliquepredictiveclusteringtrees}. Another distinct, probabilistic approach is the family of Finite Mixture of Regression (FMR) models, which assume the data arises from subpopulations each governed by its own regression model linking $\boldsymbol{x}$ to $\boldsymbol{y}$ \cite{fmr, marbac_simultaneous_2022}.

We distinguish our proposed guided clustering paradigm from standard predictive clustering, where the primary metric is prediction accuracy and low performance is typically viewed as a failure. In contrast, guided clustering (our proposal) views the guiding variable not as a target to be perfectly predicted, but as a lens to steer the discovery process. From this perspective, the goal is meaningful organization: even a weak predictor can successfully structure the latent space into clusters that are relevant to the context defined by $\boldsymbol{y}$.

Despite their conceptual differences, these methods have a common limitation: they operate directly in the raw feature space $\boldsymbol{x}$. This direct approach struggles with the curse of dimensionality and fails to capture complex non-linearities. Specifically, PCTs are constrained by greedy, axis-aligned splits, while FMRs are limited by strong parametric assumptions that are often too rigid for real-world data. 
This challenge motivates a paradigm shift towards deep representation learning, mapping data $\boldsymbol{x}$ to a compact latent space $\boldsymbol{z}$ to overcome the limitations of operating in the raw data space. More specifically, we adopt a deep generative modeling perspective. The central idea is to discover latent clusters in $\boldsymbol{x}$ that are generative of the target $\boldsymbol{y}$. This requires a framework capable of learning a compact, non-linear representation of the data that is structured for this generative-predictive task. The Variational Autoencoder (VAE) \cite{vae} serves as a cornerstone for such representation learning, inspiring a diverse range of models.
Models like VaDE \cite{VADE} and GMVAE \cite{GMVAE} learn latent representations for clustering. However, their objective is to reconstruct $\boldsymbol{x}$, meaning the latent space is structured to preserve information about $\boldsymbol{x}$ alone. Architectures like Conditional VAEs (CVAEs) \cite{CVAE}, Multimodal VAEs (MVAEs) \cite{mvae}, and Characteristic Capturing VAEs (CCVAEs) \cite{ccvae} successfully incorporate external information; however, it is used to condition the generation of $\boldsymbol{x}$ or to learn a joint space across modalities.

No existing paradigm fully synthesizes deep representation learning with the specific guided clustering goal. To bridge this gap, we draw inspiration from the Variational Information Bottleneck (VIB) \cite{variationalinformationbottleneck}. This framework provides a principled way to extract a representation that is maximally relevant to a target variable while filtering out unrelated noise. However, VIB views the latent space as a compressed encoding computed downstream from the input. Our approach adapts this information-theoretic balance into a deep generative framework. Instead of merely filtering the input, we assume the latent clusters are the pre-existing underlying structure that generates the guiding variable. This shift allows us to use the input not just to predict a target, but to uncover the latent groups that caused it. We realize this paradigm with the so-called Guided Clustering Variational AutoEncoder (GCVAE), which structures the latent space as a probabilistic mixture of components. By optimizing this generative structure to be maximally informative about the guiding variable, GCVAE forces the representation to be both predictively relevant and organized into meaningful clusters.

\section{Contribution}
Our primary contributions are:
\begin{itemize}
    \item We propose and formalize guided clustering, a framework that reframes the role of contextual information in unsupervised learning. It shifts the objective from using clusters to predict a variable, to using a variable which actively guides the discovery of meaningful data partitions.
    \item We introduce the Guided Clustering Variational Autoencoder (GCVAE), a deep generative model that operationalizes our paradigm. Its core technical innovation is the integration of a probabilistic mixture structure directly into a guided variational framework, forcing the model to learn a latent space that is both clustered and maximally informative about the guiding variable.
    \item Empirical validation on open-source and private real-world datasets, demonstrating that our model discovers coherent clusters, even in high-dimensional and non-linear settings.
\end{itemize}

The remainder of the paper is organized as follows. Section \ref{sec:model} details the generative process, inference model, and objective function of GCVAE. In Section \ref{sec:experiments}, we present our comprehensive experimental evaluation, and finally, we conclude in Section \ref{sec:ccl}.

\section{The variational implementation of the model}\label{sec:model}

To operationalize the guided clustering paradigm, we propose the Guided Clustering Variational Autoencoder (GCVAE), a deep generative model designed to discover a latent cluster structure in data $\boldsymbol{x}$ that is maximally informative for a guiding variable $\boldsymbol{y}$.

The model architecture is built on two core principles. First, it uses an inference model (encoder) to learn a compressed latent representation $\boldsymbol{z}$ of the input $\boldsymbol{x}$. This representation is optimized to act as an information bottleneck, retaining only the information from $\boldsymbol{x}$ that is necessary for the second component: a generative model (decoder) that predicts the guiding variable $\boldsymbol{y}$ from $\boldsymbol{z}$.

Crucially, the latent space is not unstructured. We model its distribution as a Gaussian Mixture Model (GMM). This imposes a distinct cluster structure, forcing the informative bottleneck $\boldsymbol{z}$ to be organized into a discrete mixture of components $\boldsymbol{c}$.

Formally, for an input $\mathbf{x} = (\boldsymbol{x}_1, ..., \boldsymbol{x}_n)$, with $\boldsymbol{x}_i \in \mathbb{R}^{d_x}$, and a corresponding guiding variable $\mathbf{y} = (\boldsymbol{y}_1,...,\boldsymbol{y}_n)$, with $\boldsymbol{y}_i \in \mathbb{R}^{d_y}$, our model learns a continuous latent variable $\mathbf{z}=(\boldsymbol{z}_1,...,\boldsymbol{z}_n)$, with $\boldsymbol{z}_i \in \mathbb{R}^J$, and infers a discrete cluster assignment $\mathbf{c}=(c_1,...,c_n)$ with $c_i \in \{1,...,K\}$, with $n$ the number of observations. We assume that $(\boldsymbol{x}_i,\boldsymbol{y}_i,\boldsymbol{z}_i,c_i)$ are independent and identically distributed. For simplicity, we will omit the index $i$ in the rest of this paper when there is no ambiguity. \\
Figure~\ref{fig:schema_overview} provides a graphical overview of this architecture. The following subsections detail the probabilistic formulation of the generative process, the inference model, and the final training objective.

\begin{figure}[ht]
    \begin{center}
     \includegraphics[width=0.65\textwidth]{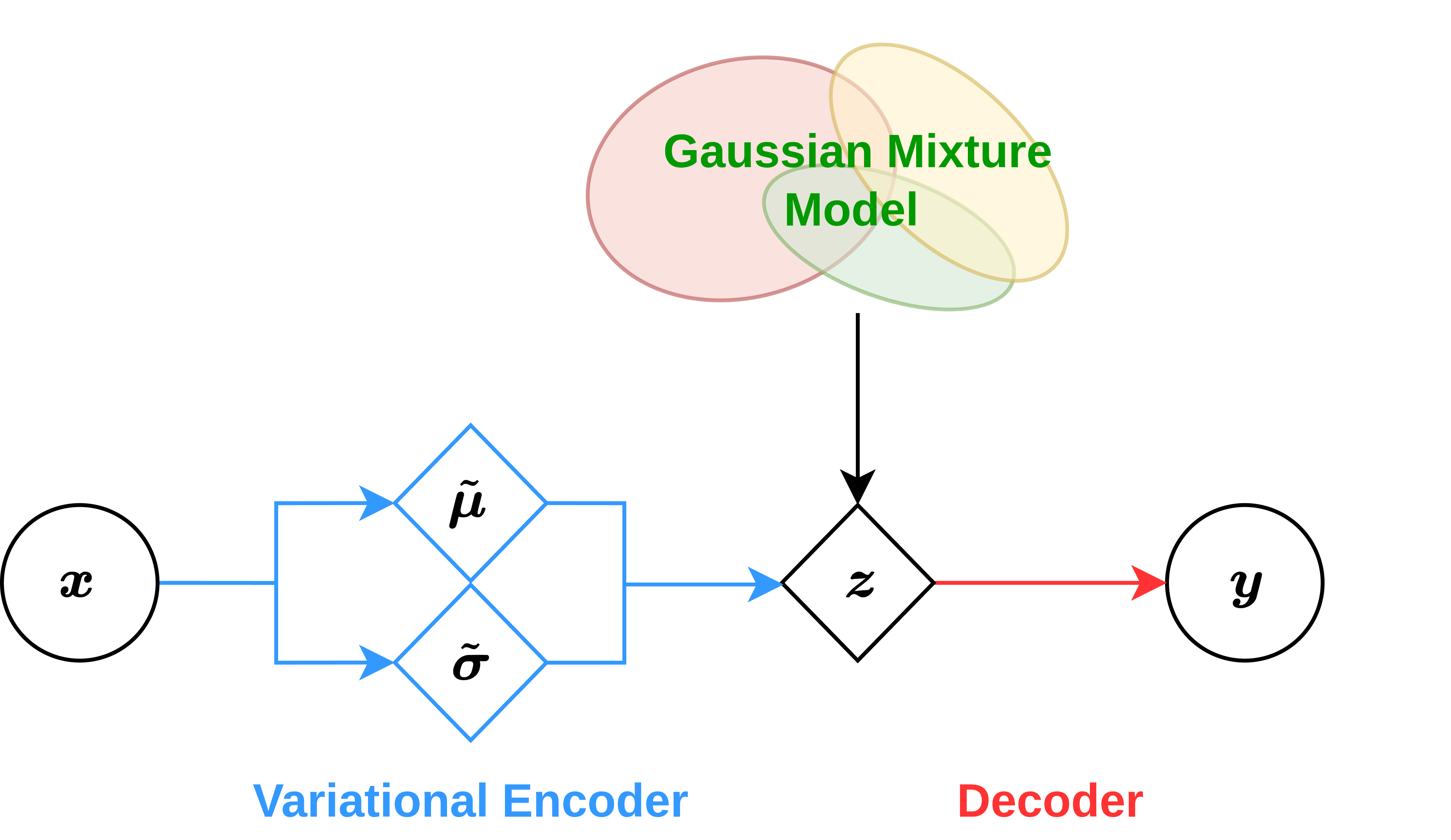}
     \caption{Overview of the GCVAE model. \textit{Diamond-shaped units denote latent variables, while round ones denote observations.}} \label{fig:schema_overview}
    \end{center}
\end{figure}   

\subsection{Generative model}

The generative process of our model consists of the following steps:
\begin{enumerate}
    \item Choose a cluster $c$ using a Multinomial distribution: 
    \begin{align}
        p_{\boldsymbol{\pi}}(c) = Cat(c;\boldsymbol{\pi}).
    \end{align}
    \item Generate a latent vector $\boldsymbol{z}$ conditioned on the cluster $c$ using a spherical Gaussian distribution: 
    \begin{align}
        p_{\boldsymbol{\mu}_c, \boldsymbol{\sigma}_c}(\boldsymbol{z}|c) = \mathcal{N}(\boldsymbol{z}; \boldsymbol{\mu}_{c}, \boldsymbol{\sigma}_{c}^2 I ). \label{eq_gmm}
    \end{align}
    \item Generate the variable $\boldsymbol{y}$ from the latent vector $\boldsymbol{z}$: 
    \begin{align}
        p_{{\boldsymbol{\theta}}}(\boldsymbol{y}|\boldsymbol{z}) = \mathcal{N}(\boldsymbol{y}; f_{{\boldsymbol{\theta}}}(\boldsymbol{z}), \boldsymbol{I}). \label{eq_f_y}
    \end{align}
\end{enumerate}

Here, $\boldsymbol{\pi} = (\pi_1,...,\pi_K) \in [0,1]^K$, with $\pi_c$ the prior probability for cluster $c$, $\sum_{c=1}^K \pi_c = 1$, $\boldsymbol{\mu}_c = \{\mu_{c,j}\}_{j=1,...,J}$ and $\boldsymbol{\sigma}_c= \{\sigma_{c,j}\}_{j=1,...,J}$ are respectively the mean and the diagonal covariance of the multivariate normal distribution corresponding to cluster $\boldsymbol{c}$, $J$ is the dimension of the latent space, $\boldsymbol{I}$ is an identity matrix, $f_{{\boldsymbol{\theta}}}(\boldsymbol{z})$ 
 is a network with input $\boldsymbol{z}$ and parametrized by ${\boldsymbol{\theta}}$.
 
\begin{figure}[ht]
    \centering
    \begin{minipage}[b]{0.3\textwidth}
        \centering
        \includegraphics[width=0.8\textwidth]{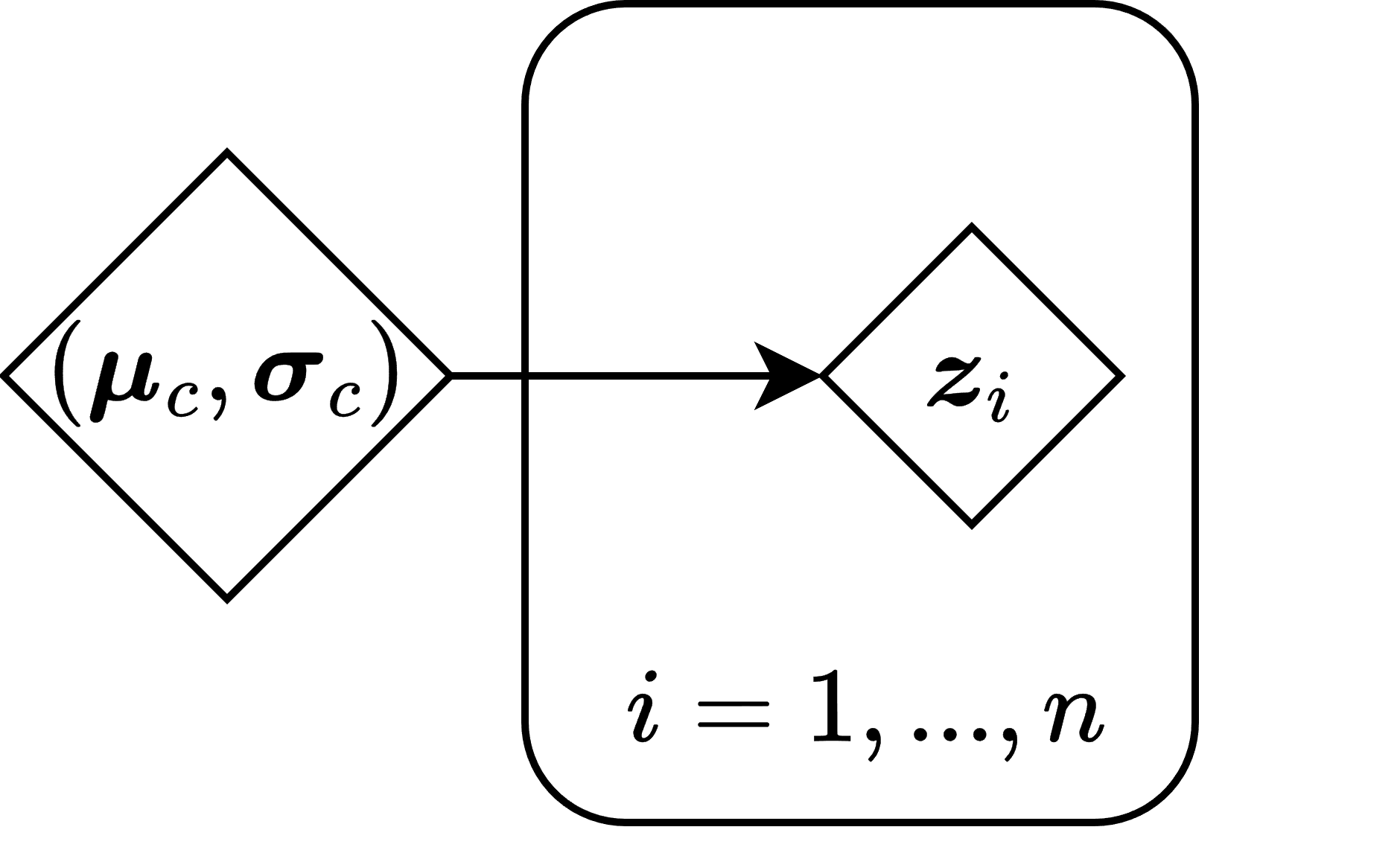}
    \end{minipage}
    \begin{minipage}[b]{0.3\textwidth}
        \centering
        \includegraphics[width=0.85\textwidth]{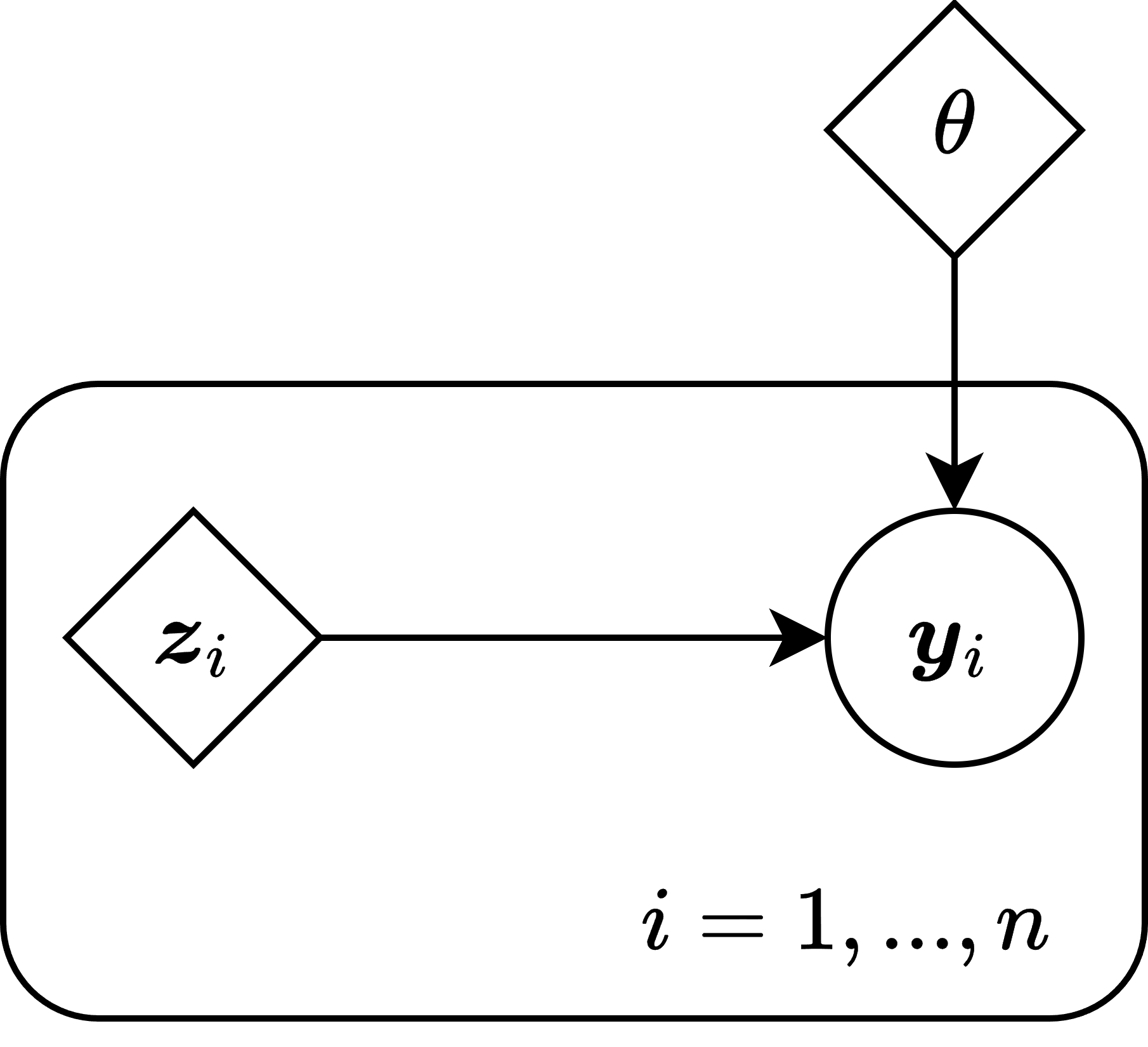}
    \end{minipage}
    \caption{A graphical representation of the generative process, with the GMM in Equation~\ref{eq_gmm} (left) and the decoder in Equation~\ref{eq_f_y} (right). }
    \label{fig:three_figures}
\end{figure}

Following the structure of established deep generative clustering models  \cite{VADE, GMVAE}, we assume: $p_{{\boldsymbol{\theta}}}(\boldsymbol{y}|\boldsymbol{z}, c) = p_{{\boldsymbol{\theta}}}(\boldsymbol{y}|\boldsymbol{z})$.
According to the generative process above, illustrated in Figure~\ref{fig:three_figures}, the joint probability of the model is given by:
\begin{align}
    p_{\boldsymbol{\Theta}}(\mathbf{z},\mathbf{c},\mathbf{y}) = p_{\boldsymbol{\Theta}}(\mathbf{y}|\mathbf{z}) p_{\boldsymbol{\Theta}}(\mathbf{z}|\mathbf{c}) p_{\boldsymbol{\Theta}}(\mathbf{c})  \label{eq:generative}
\end{align}
with $\boldsymbol{\Theta} = \{ \boldsymbol{\boldsymbol{\mu}}_1, ..., \boldsymbol{\boldsymbol{\mu}}_K, \boldsymbol{\sigma}_1, ..., \boldsymbol{\sigma}_K, \pi_1, ..., \pi_K, {\boldsymbol{\theta}}\}$.

\subsection{Inference model}\label{subsec:inference}

To estimate the generative parameters $\boldsymbol{\Theta}$, a standard VAE procedure would typically infers latent variables from the same data it intends to reconstruct, seeking the posterior $p_{\Theta}(\mathbf{z}, \mathbf{c} | \mathbf{y})$. In contrast, the guided clustering paradigm is built on an informational asymmetry: the encoder must learn to map the input features $\mathbf{x}$ to a latent space that mimics the structure of the ideal posterior conditioned on $\mathbf{y}$, as illustrated in Figure~\ref{fig:schema_overview}.

We therefore define the inference task as finding the variational approximation $q_{{\boldsymbol{\phi}}}(\mathbf{z}, \mathbf{c} | \mathbf{x})$ \cite{review_VI} that best captures this ideal structure. Ideally, we would minimize the Kullback-Leibler divergence between our variational approximation and this target posterior, which gives us the central hypothesis of the guided clustering:
\begin{align}
    \min_{{\boldsymbol{\phi}}} \text{KL}\left[ q_{{\boldsymbol{\phi}}}(\mathbf{z}, \mathbf{c} | \mathbf{x}) \,||\, p_{\boldsymbol{\Theta}}(\mathbf{z}, \mathbf{c} | \mathbf{y}) \right].\label{eq:hyp_GC}
\end{align}


For the implementation, we illustrate the encoder structure in Figure~\ref{fig:schema_encoder}. We assume a classical mean-field approximation for the variational distribution:
\begin{align}
    q_{\boldsymbol{\phi}}(\mathbf{z},\mathbf{c}|\mathbf{x}) = q_{\boldsymbol{\phi}}(\mathbf{z}|\mathbf{x}) q(\mathbf{c}|\mathbf{x}) \label{eq:inference}
\end{align}
with
\begin{gather}
    q_{\boldsymbol{\phi}}(\boldsymbol{z}_i|\boldsymbol{x}_i) = \mathcal{N}(\boldsymbol{z}_i;\tilde{\boldsymbol{\boldsymbol{\mu}}}_i, \tilde{\boldsymbol{\sigma}}_i^2 I) , \\
    [\tilde{\boldsymbol{\boldsymbol{\mu}}}_i, \log \tilde{\boldsymbol{\sigma}}_i] = g_{{\boldsymbol{\phi}}}(\boldsymbol{x}_i) \label{eq_g}   
\end{gather}
where $\tilde{\boldsymbol{\boldsymbol{\mu}}}_i \in \mathbb{R}^J$, $\tilde{\boldsymbol{\sigma}}_i \in \mathbb{R}^{J+}$ for $i \in \{1,..., n\}$. 
Here, $g_{{\boldsymbol{\phi}}}$ is a neural network parameterized by ${\boldsymbol{\phi}}$, and the detail of the estimation of $q(\boldsymbol{c}|\boldsymbol{x})$ is described in Section \ref{subsec:objective}.

\begin{figure}[ht]
\begin{center}
 \includegraphics[width=0.35\textwidth]{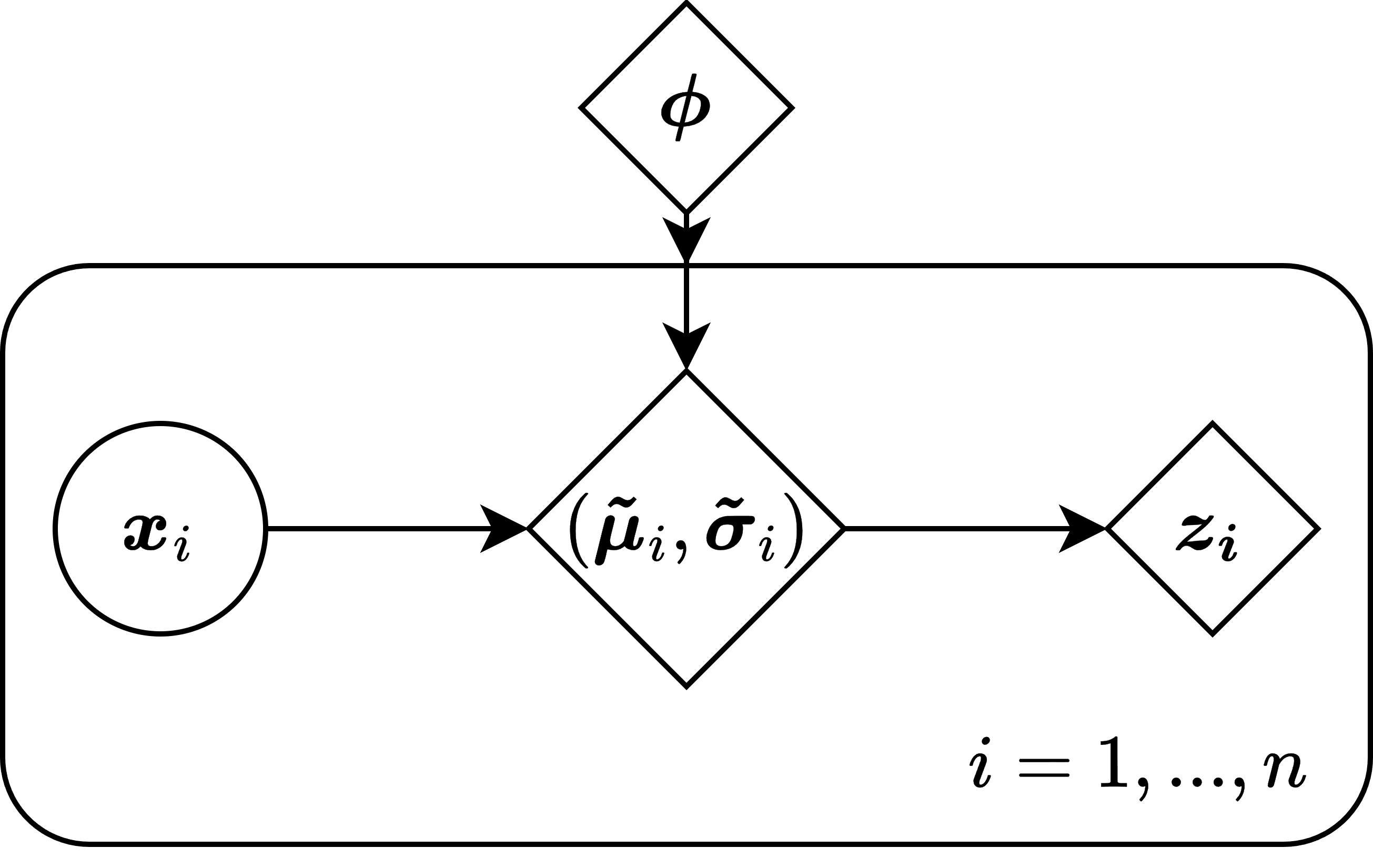}
 \caption{A graphical representation of the inference model.} \label{fig:schema_encoder}
\end{center}
\end{figure}

\subsection{Estimation loss}\label{subsec:objective}

Minimizing the divergence in Equation \ref{eq:hyp_GC} is equivalent to maximizing the Evidence Lower Bound (ELBO):
\begin{align*}
    &\text{KL}\left[ q_{{\boldsymbol{\phi}}}(\mathbf{z}, \mathbf{c} | \mathbf{x}) \,||\, p_{\boldsymbol{\Theta}}(\mathbf{z}, \mathbf{c} | \mathbf{y}) \right] 
    = \text{KL}\left[ q_{{\boldsymbol{\phi}}}(\mathbf{z}, \mathbf{c} | \mathbf{x}) \,||\, p_{\boldsymbol{\Theta}}(\mathbf{z}, \mathbf{c}) \right] - \mathbb{E}_{q_{\boldsymbol{\phi}}(\mathbf{z},\mathbf{c}|\mathbf{x})}[\log p_{\boldsymbol{\Theta}}(\mathbf{y}|\mathbf{z}, \mathbf{c})] + \log p_{\boldsymbol{\Theta}}(\mathbf{y}) \\
    & \Rightarrow \log p_{\boldsymbol{\Theta}}(\mathbf{y}) \geq \text{ELBO}_{\boldsymbol{\Theta}, {\boldsymbol{\phi}}}(\mathbf{x},\mathbf{y},\mathbf{z},\mathbf{c})  = \mathbb{E}_{q_{\boldsymbol{\phi}}(\mathbf{z},\mathbf{c}|\mathbf{x})}[\log p_{\boldsymbol{\Theta}}(\mathbf{y}|\mathbf{z}, \mathbf{c})]  - \text{KL}\left[ q_{{\boldsymbol{\phi}}}(\mathbf{z}, \mathbf{c} | \mathbf{x}) \,||\, p_{\boldsymbol{\Theta}}(\mathbf{z}, \mathbf{c}) \right].
\end{align*}

\paragraph{Remark} Minimizing the divergence \ref{eq:hyp_GC} thus compels the inference model to satisfy two competing goals: it must organize the data $\mathbf{x}$ into the structural constraints of the mixture prior $p_{\boldsymbol{\Theta}}(\mathbf{z}, \mathbf{c})$, while simultaneously retaining sufficient information to reconstruct the guiding variable $\mathbf{y}$.

While this standard formulation provides a rigorous lower bound on the log-likelihood, we want a more flexible control over the information flow. Specifically, we need to regulate how much information from $\mathbf{x}$ is compressed into the latent structure $(\mathbf{z}, \mathbf{c})$ versus how much is used to predict $\mathbf{y}$. 

To control this trade-off, we frame the learning problem as a constrained optimization task, following the $\beta$-VAE framework \cite{betavae}. We condition the inference model on $\mathbf{x}$ while targeting the reconstruction of $\mathbf{y}$. The goal is to maximize the reconstruction quality of $\mathbf{y}$ while constraining the information capacity of the latent representation learned from $\mathbf{x}$, using Kullback–Leibler divergence (KL-divergence).  With $\epsilon_\beta$ specifying the strength of the applied constraint, this can be formally stated as: 
\begin{align*}
\max_{\boldsymbol{\Theta}, {\boldsymbol{\phi}}} \ & \mathbb{E}_{q_{\boldsymbol{\phi}}(\mathbf{z},\mathbf{c}|\mathbf{x})}[\log p_{\boldsymbol{\Theta}}(\mathbf{y}|\mathbf{z},\mathbf{c})] \text{ subject to } \text{KL}[q_{\boldsymbol{\phi}}(\mathbf{z},\mathbf{c}|\mathbf{x})||p_{\boldsymbol{\Theta}}(\mathbf{z},\mathbf{c})] < \epsilon_\beta.
\end{align*}

Rewriting it as a Lagrangian under the KKT conditions \cite{kkt}, we obtain: 
\begin{align*}
    \mathcal{F}_{\boldsymbol{\Theta}, {\boldsymbol{\phi}}, \beta}(\mathbf{x},\mathbf{y},\mathbf{z},\mathbf{c})  
    &=  \mathbb{E}_{q_{\boldsymbol{\phi}}(\mathbf{z},\mathbf{c}|\mathbf{x})}[\log p_{\boldsymbol{\Theta}}(\mathbf{y}|\mathbf{z},\mathbf{c})]   
     \ - \beta \ (\text{KL}[q_{\boldsymbol{\phi}}(\mathbf{z},\mathbf{c}|\mathbf{x})||p_{\boldsymbol{\Theta}}(\mathbf{z},\mathbf{c})] - \epsilon_\beta) 
\end{align*}
where the KKT multiplier $\beta$ is the regularisation coefficient.

Still following the $\beta$-VAE article principle and formulation, we can derive the final ELBO:
\begin{align*}
    \mathcal{F}_{\boldsymbol{\Theta}, {\boldsymbol{\phi}}, \beta}(\mathbf{x},\mathbf{y},\mathbf{z},\mathbf{c}&) 
    \geq \mathcal{L}_{\boldsymbol{\Theta}, {\boldsymbol{\phi}},\beta}(\mathbf{x},\mathbf{y}, \mathbf{z},\mathbf{c}) 
     = \mathbb{E}_{q_{\boldsymbol{\phi}}(\mathbf{z},\mathbf{c}|\mathbf{x})}[\log p_{\boldsymbol{\theta}}(\mathbf{y}|\mathbf{z},\mathbf{c})] - \beta \ \text{KL}[q_{\boldsymbol{\phi}}(\mathbf{z},\mathbf{c}|\mathbf{x})||p_{\boldsymbol{\Theta}}(\mathbf{z},\mathbf{c})] .
\end{align*}

By adjusting $\beta$, we can prioritize the discovery of a structured latent space that is not just a passthrough for the data, but an effective information bottleneck. Importantly, the introduction of this weighting factor does not impact the theoretical convergence properties discussed in Section~\ref{sec:consistency}.

The ELBO can then be developed as below (cf. Appendix \ref{elbo}):
\begin{align*}
\mathcal{L}_{\boldsymbol{\Theta}, {\boldsymbol{\phi}},\beta}(\mathbf{x},\mathbf{y}, \mathbf{z},\mathbf{c})
\simeq & -  \frac{1}{L} \sum_{l=1}^L \sum_{i=1}^n||\boldsymbol{y}_i - f_{{\boldsymbol{\theta}}}(\boldsymbol{z}_i^{(l)})||_2^2  \\
& + \beta \sum_{i=1}^n \left( - \frac{1}{2} \sum_{c=1}^K q(c|\boldsymbol{x}_i) 
 \sum_{j=1}^J \left(\log \sigma_{c,j}^2 
+ \frac{\tilde{\sigma}_{i,j}^2}{\sigma_{c,j}^2}  \right. \right.
 \left. + \frac{\left(\tilde{\mu}_{i,j}-\mu_{c,j}\right)^2}{\sigma_{c,j}^2} \right) \\
 & \ \ +\sum_{c=1}^K q(c|\boldsymbol{x}_i) \log \pi_{c} 
 \left. + \frac{1}{2} \sum_{j=1}^J(1+\log \tilde{\sigma}_{i,j}^2)  - \sum_{c=1}^K q(c|\boldsymbol{x}_i) \log q(c|\boldsymbol{x}_i) \right)
\end{align*}

with $\boldsymbol{z}_i^{(l)} \sim \mathcal{N}(\tilde{\boldsymbol{\mu}}_i, \tilde{\boldsymbol{\sigma}}_i^2I)$, and $[\tilde{\boldsymbol{\mu}}_i, \log \tilde{\boldsymbol{\sigma}}_i^2] = g_{{\boldsymbol{\phi}}}(\boldsymbol{x}_i)$ with $\tilde{\boldsymbol{\mu}}_i = \{\tilde{\mu}_{i,j}\}_{j=1,...,J}$ and $\tilde{\boldsymbol{\sigma}}_i = \{\tilde{\sigma}_{i,j}\}_{j=1,...,J}$. Recall that $J$ is the dimension of $\boldsymbol{z}_i$.

We finally approximate $q(c|\boldsymbol{x})$ using the SGVB estimator (cf. Appendix \ref{approx}), where $L$ is the number of Monte Carlo samples: 
\begin{align*}
    q(c|\boldsymbol{x}) 
    &= \mathbb{E}_{q_{\boldsymbol{\phi}}(\boldsymbol{z}|\boldsymbol{x})} [p_{\boldsymbol{\Theta}}(c|\boldsymbol{z})] 
    \simeq \frac{1}{L} \sum_{l=1}^L \frac{ p_{\boldsymbol{\Theta}}(\boldsymbol{z}^{(l)}|c)p_{\boldsymbol{\Theta}}(c)}{\sum_{c'=1}^K p_{\boldsymbol{\Theta}}(\boldsymbol{z}^{(l)}|c')p_{\boldsymbol{\Theta}}(c')}.
\end{align*}

To encourage sharp clusters assignment while maintaining the differentiability of the model, we employ the Gumbel-Softmax reparameterization trick. This provides a differentiable approximation to sampling from the categorical distribution $q(c|\boldsymbol{x})$, thereby ensuring that gradients can flow through the entire model. 

\subsection{Consistency}\label{sec:consistency}

To establish the theoretical validity of our estimator, we analyze the convergence properties of the GCVAE. We utilize the framework of Generalized Variational Inference (GVI) \cite{knoblauch20193argsgeneralized}, which generalizes standard Bayesian inference to posterior beliefs derived from arbitrary loss functions. Our analysis relies on the frequentist consistency results established in Theorem 2 of 
\cite{knoblauch2019frequentist}. By applying this framework to the specific structure of the GCVAE, we can state the following: 

\begin{proposition}[GCVAE consistency]
    Let $\boldsymbol{\psi} = \{\boldsymbol{\Theta}, {\boldsymbol{\phi}}\} \in \boldsymbol{\Psi}$ denote the complete set of learnable global parameters in the GCVAE model. If the data are i.i.d. and $\boldsymbol{\Psi}$ is compact, the estimator $\hat{\boldsymbol{\psi}}_n$ converges to a point mass at the population-optimal parameter $\boldsymbol{\psi}^*$ as number of observations $n \to \infty$.
\end{proposition}

The proof is available in \cite{knoblauch2019frequentist}. The direct application of this theorem is non-trivial due to the non-standard nature of our mixed variational family, which involves both continuous and discrete latent variables. As the rigorous verification of these conditions is rarely detailed in the literature, we provide a full verification in Appendix~\ref{appendix:cv}.


\section{Experiments}\label{sec:experiments}
In this section, we evaluate the performance of the GCVAE on two datasets. The model has been implemented in Python using Pytorch \cite{pytorch}, and the code is available at https://github.com/vcourrier/gcvae.

\subsection{Practical implementation guidelines}

To effectively apply the GCVAE in practice, specific attention must be paid to the architectural choices and the optimization of the ELBO. While the generative framework is agnostic to the specific family of neural networks used, the choice of hyperparameters impacts the model's ability to uncover meaningful cluster structures.

\paragraph{Initializating GMM}
In this work, pre-training is used to initialize GMM parameters $(\pi_c, \boldsymbol{\mu}_c, \boldsymbol{\sigma}_c)$, a common practice in deep clustering \cite{VADE, dec}. We pretrain the model without the clusters in the latent space, leading to the following ELBO:
\begin{align*}
    \mathcal{L}^{\text{pretrain}}_{\boldsymbol{\Theta}, {\boldsymbol{\phi}}, \beta}(\mathbf{x},\mathbf{y}, \mathbf{z},\mathbf{c}) 
    = \ & \ \mathbb{E}_{q_{\boldsymbol{\phi}}(\mathbf{z}|\mathbf{x})}[\log p_{\boldsymbol{\Theta}}(\mathbf{y}|\mathbf{z})] 
     - \beta \ \text{KL}[q_{\boldsymbol{\phi}}(\mathbf{z}|\mathbf{x})||\mathcal{N}(\mathbf{z};0_J,\boldsymbol{I})]
\end{align*}
with $0_J$ a vector null of dimension $J$. \\
After pretraining the encoder and decoder for a few epochs, we can fit a GMM in the latent space to initialize its parameters.

\paragraph{Network architectures and the information bottleneck} The encoder and decoder functions ($g_{\boldsymbol{\phi}}$ and $f_{\boldsymbol{\theta}}$) can be parameterized using various architectures suited to the data modality, such as Multilayer Perceptrons (MLPs), Convolutional Neural Networks (CNNs), Recurrent Neural Networks (RNNs)... However, from a statistical perspective, we caution against the use of excessively complex networks. If the encoder capacity is too high, the network may \enquote{absorb} the data's structure into its parameters, rendering the latent space $\mathbf{z}$ a simple passthrough rather than a structured representation. To ensure the model learns robust clusters, we design the networks to act as an effective information bottleneck. By intentionally constraining the network capacity (e.g., reducing the number of layers or units), we force the model to compress the essential features of $\mathbf{x}$ into the latent clusters $\mathbf{c}$ and variable $\mathbf{z}$ to maximize the predictive accuracy of $\mathbf{y}$. This parsimonious approach ensures that the clustering structure captures the signal rather than the network weights.

\paragraph{The regularization parameter $\beta$} The coefficient $\beta$ controls the trade-off between the reconstruction fidelity of the guiding variable and the adherence of the latent space to the GMM prior. In our experiments, we observe that a small weighting ($\beta < 1$) typically yields a more tractable optimization problem \cite{burgess2018understanding}. This prevents the complex KL-divergence term associated with the GMM from dominating the loss. Multiple techniques can be used to find a suitable $\beta$, via standard grid search or KL annealing \cite{fu2019cyclical}, where $\beta$ is gradually increased during training.

\paragraph{Optimization and hyperparameters}
The model parameters are estimated using the Adam optimizer \cite{adam}, a standard for stochastic gradient-based optimization. Beyond $\beta$, the key hyperparameters requiring selection include the learning rate, the dimensions of the neural networks, and the number of latent clusters $K$. The choice of $K$ may be driven by prior domain knowledge (as seen in our clinical application) or determined via model selection criteria compatible with the VAE framework.

\subsection{MNIST-SVHN dataset}

\begin{figure}[ht]
\begin{center}
 \includegraphics[width=0.4\textwidth]{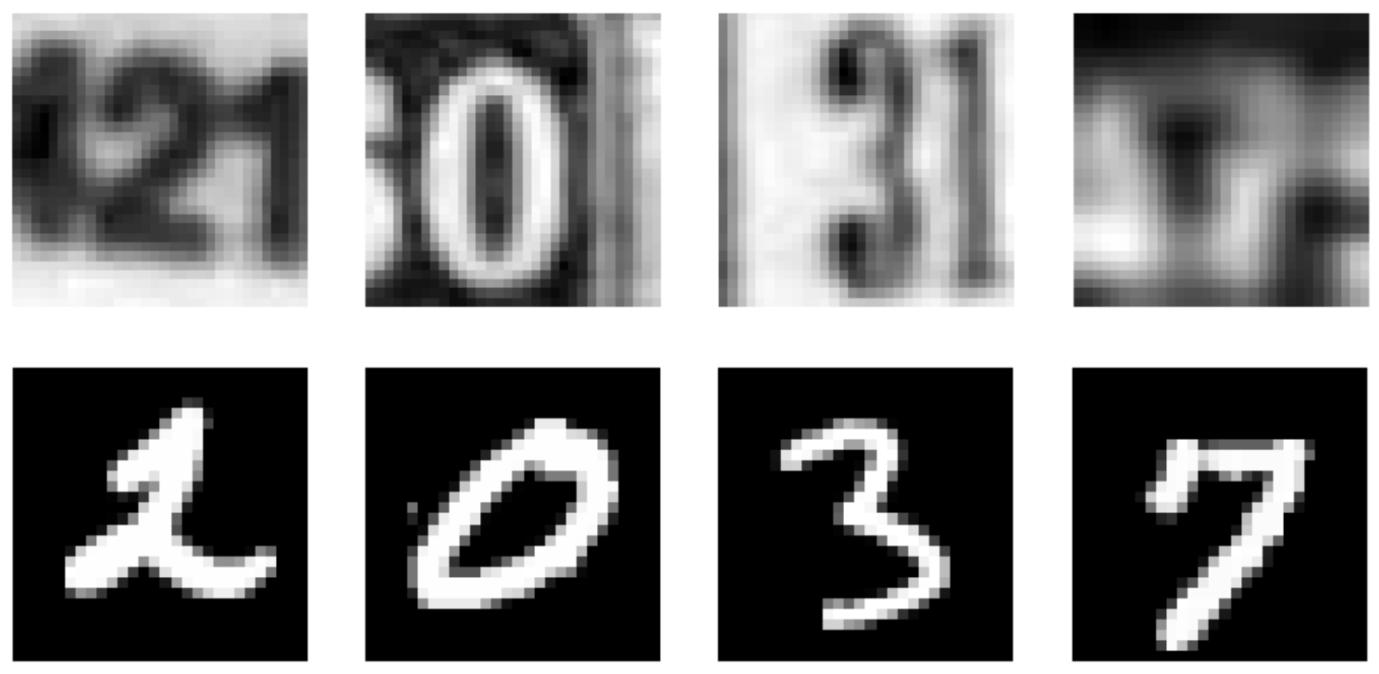}
 \caption{Example of SVHN (top) and MNIST (bottom) data.} \label{fig:mnist_svhn}
\end{center}
\end{figure}

\paragraph{Dataset} We evaluate our proposed approach on a dataset composed of paired MNIST and black-and-white SVHN images, with the SVHN image as our input $\mathbf{x}$ and the MNIST image as our guiding variable $\mathbf{y}$, where each pair represents the same digit class,  following the methodology introduced in the MMVAE model \cite{mmvae}. Each instance of a digit class (in either dataset) is randomly paired with $30$ instances of the same digit class from the other dataset. We use the standard training partitions of MNIST and SVHN (see details in the respective datasets description). As illustrated in Figure~\ref{fig:mnist_svhn}, the dataset presents a variety of styles, making the task of capturing the digit classes challenging. 

\paragraph{Implementation details} For this experiment, we use CNNs for the encoder of SVHN, and a MLP for the decoder of MNIST. The latent space dimension is set to 20. For learning, we use the Adam optimizer \cite{adam}. We set $\beta = 0.1$. During the pre-training, we use a learning rate of $0.0001$ and run it for 5 epochs. During the training, we use a learning rate of $0.0005$ for the parameters of the encoder and decoder, and $0.0005$ too for the parameters of the GMM, and set the number of clusters to $10$, run it for 50 epochs. 

\paragraph{Impact of the guiding variable} To assess the influence of the predictive variable on our model's clustering performance, we conduct experiments on a multitude of clustering methods. 

Consistent with prior research in image clustering, we evaluate performance using global classification accuracy (ACC), where a cluster-to-class mapping is determined via the Hungarian algorithm \cite{Hungarian}. Without a guiding variable $\boldsymbol{y}$, with the same model but reconstructing $\boldsymbol{x}$, the model achieves an ACC of $11.6$\% (average on 10 runs). In contrast, using the MNIST data as the guiding variable significantly improves the ACC to $62.1$\% (average on 10 runs). 

\begin{table}[ht]
\caption{ACC on standard clustering benchmarks.}
\label{tab:acc_comparaison}
\vskip 0.15in
\begin{center}
\begin{small}
\begin{sc}
\begin{tabular}{lcr}
\toprule
 Model & ACC \\
\midrule
\multicolumn{2}{l}{\emph{Clustering models with image-specific transformations}} \\
\hspace{1em} DTI K-means     \cite{dti}            & 44.5\% \\
\hspace{1em} SCAE            \cite{scae}           & 55.3\% \\
\hspace{1em} DTI GMM         \cite{dti}            & 57.4\% \\
\hspace{1em} ACOL-GAR        \cite{acol_gar}       & 76.8\% \\
\hdashline 
\multicolumn{2}{l}{\emph{Clustering models with domain-agnostic designs}} \\
\hspace{1em} GMM             \cite{gmm}            & 11.6\% \\
\hspace{1em} DEC             \cite{dec}            & 11.9\% \\
\hspace{1em} K-means         \cite{kmeans}         & 12.2\% \\
\hspace{1em} DeepCluster-v2  \cite{deepclusterv2}  & 20.6\% \\
\hspace{1em} VaDE            \cite{VADE}           & 30.8\% \\
\hspace{1em} MFCVAE          \cite{mfcvae}         & 56.3\% \\
\hspace{1em} IMSAT           \cite{imsat}          & 57.3\% \\
\hspace{1em} \textbf{GCVAE} (our proposal)         & 62.1\% \\
\bottomrule
\end{tabular}
\end{sc}
\end{small}
\end{center}
\vskip -0.1in
\end{table}

\paragraph{Comparison on standard benchmarks} 

To quantify the value added by the guiding variable $\mathbf{y}$, we contrast the performance of GCVAE against established unsupervised clustering methods. We categorize these baselines into two groups based on their reliance on domain-specific knowledge (Table \ref{tab:acc_comparaison}).

In the first category, models rely on image-specific transformations or architectures, effectively acting as implicit guidance. DTI \cite{dti} incorporates spatial or morphological transformations, SCAE \cite{scae} segments images into part templates before reasoning about \enquote{object capsules}, and ACOL-GAR \cite{acol_gar} applies domain-specific transformations to generate pseudo parent classes and achieves high performance. Notably, ACOL-GAR achieves the highest performance on this benchmark (76.8\%), surpassing GCVAE (62.1\%). However, this superiority stems from \enquote{hard-coded} domain knowledge: the method relies on invariances specific to visual data to generate supervision. While highly effective for images, this reliance makes such methods less transferable to non-visual tasks (e.g., tabular or sensor data) where such domain-specific invariants are unknown or undefined.

In the second category, domain-agnostic methods seek broader applicability without depending on extensive image transformations. This includes classical approaches like GMM \cite{gmm} and K-means \cite{kmeans}, as well as deep methods like DEC \cite{dec}, VaDE \cite{VADE}, MFCVAE \cite{mfcvae}, IMSAT \cite{imsat}, \cite{deepcluster} and DeepCluster-v2 \cite{deepclusterv2}. As shown in Table \ref{tab:acc_comparaison}, these methods struggle to recover the digit classes solely from raw pixel statistics, with accuracies ranging from 11.6\% to 57.3\%. GCVAE achieves 62.1\% ACC, surpassing these domain-agnostic models by leveraging the guiding variable exclusively in the generative process.

This comparison highlights the specific niche of GCVAE: while it may not outperform specialized models equipped with extensive domain-specific augmentations (like ACOL-GAR), it significantly outperforms generic unsupervised methods by effectively leveraging the guiding variable $\boldsymbol{y}$. This offers a flexible solution that provides structure to complex data without requiring the manual design of domain-specific transformations.

\begin{figure}[ht]
\begin{center}
 \includegraphics[width=0.6\textwidth]{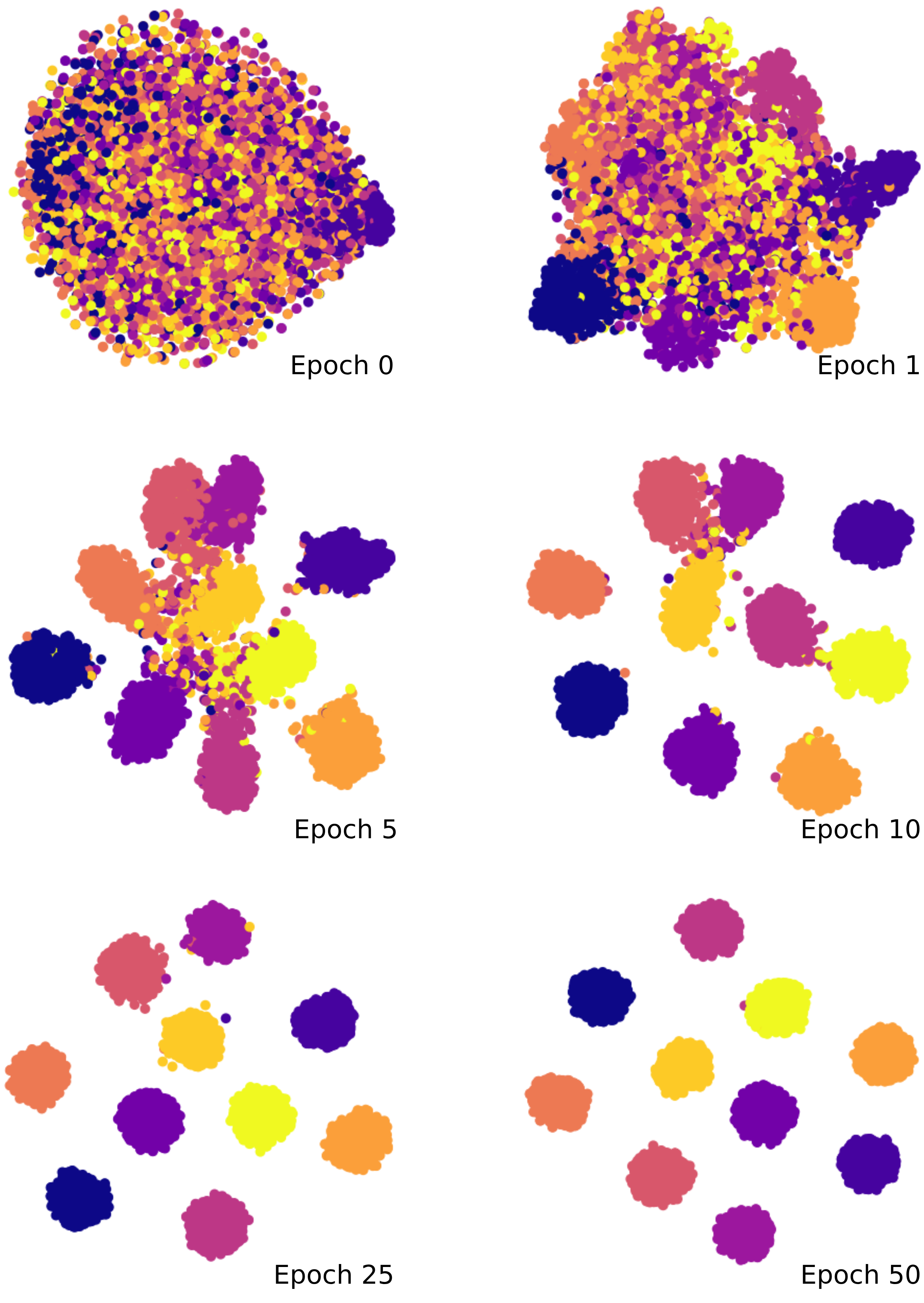}
 \caption{t-SNE visualisations at different epoch during the training of $5,000$ training images.} \label{fig:tsne}
\end{center}
\end{figure}

\paragraph{Visualization of the clusters during the training} Figure~\ref{fig:tsne} shows the t-SNE \cite{tsne} visualizations of our model’s latent space at different training epochs ($1$, $5$, $10$, $25$, $50$) for $5,000$ SVHN training examples. Each point in the plot corresponds to the latent vector of a single image, colored according to its ground-truth label. Initially, after the pretraining step, the latent-space representation is relatively unstructured. By epoch $5$, clusters begin to appear, indicating that our model is starting to learn features that help separate different digit classes. As training continues, these clusters become increasingly well-defined, and data points sharing the same label gather into tighter, more distinct regions. Notably, the boundaries between clusters also grow clearer, suggesting that the learned representations reflect class-specific properties more effectively over time. Overall, this progressive separation of clusters highlights how the guided training procedure refines the latent-space vectors over time to achieve better discrimination among different labels in an unsupervised setting.

\subsection{Sleep dataset}\label{subsec:sleep_dataset}

\paragraph{Dataset} To demonstrate the operational utility of our guided clustering paradigm, we apply GCVAE to a real-life dataset derived from Withings' proprietary data. Withings is a French company in the field of digital health, known for designing and manufacturing a wide range of connected health devices\footnote{https://www.withings.com/us/en/}. The experiment is designed to demonstrate how a guiding variable can help recover meaningful and coherent subgroups from complex data. While this dataset allows us to showcase our method's properties, the discovered subgroups are intended as a methodological proof-of-concept rather than a definitive clinical finding. Furthermore, the dataset has not undergone rigorous debiasing procedures.

Our analyses are based on sleep data collected from $50,000$ individuals, with each contributing one night of data. All personally identifiable information has been removed in compliance with GDPR guidelines. We use seven features for the input vector $\boldsymbol{x}$: \textit{sleep\_duration}, \textit{bmi}, \textit{age}, \textit{light\_sleep\_duration}, \textit{deep\_sleep\_duration}, \textit{nb\_sleep\_interruptions}, and \textit{avg\_night\_hr}. The guiding variable $\boldsymbol{y}$ explored is the \textit{apnea\_hypopnea\_index}, a standard clinical metric that quantifies sleep apnea severity by measuring the number of breathing interruptions per hour of sleep. A detailed overview of the dataset, including feature descriptions and Apnea-Hypopnea Index (AHI) categorization, is provided in Appendix~\ref{appendix:dataset}. We split the dataset into train, test, and validation sets of repartition $70$\%, $20$\% and $10$\% respectively.

\paragraph{Implementation details} The objective is to demonstrate that using the Apnea-Hypopnea Index (AHI) as a guiding variable $\boldsymbol{y}$ allows the model to discover more clinically coherent partitions of the user data $\boldsymbol{x}$ than an unguided approach. We set $K=3$ clusters to search for distinct user phenotypes.

For a principled comparison, we contrast our GCVAE (which learns a mapping $\boldsymbol{x} \mapsto y$) with a unguided baseline. This baseline uses the same architecture but is adapted for a different task. It receives the concatenated input $(\boldsymbol{x},\boldsymbol{y})$ and is trained to find a clustered representation that reconstructs both variables. This comparison allows us to differentiate between finding structure in the joint $(\boldsymbol{x},\boldsymbol{y})$ space versus finding structure within $\boldsymbol{x}$ that is relevant to $\boldsymbol{y}$.

For this experiment, we use MLPS in the encoder and the decoder. For learning, we use the Adam optimizer, as before. During the pre-training ($5$ epochs), we use a learning rate of $0.0005$ and a $\beta$ of $0.001$. During the training ($50$ epochs), we use a learning rate of $0.0001$ for the parameters of the encoder and decoder, and $0.00001$ for the parameters of the GMM, and a $\beta$ of $0.01$. %

\begin{figure}[ht]
\begin{center}
 \includegraphics[width=0.6\textwidth]{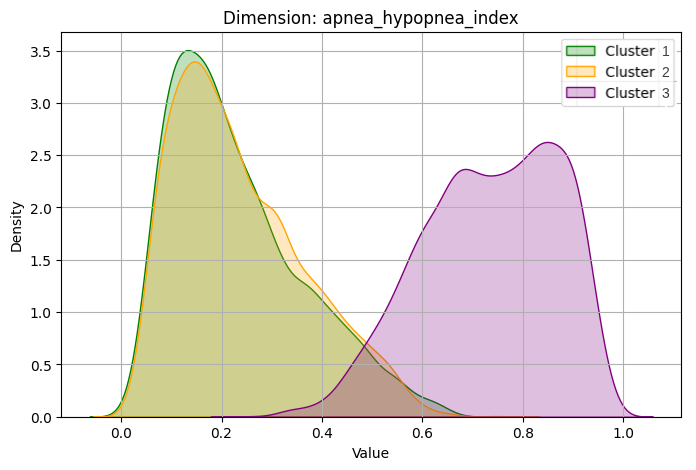}
 \caption{Distribution of the AHI within each cluster discovered by the unguided model. Note the significant overlap between the distributions for Cluster 1 and Cluster 2.} \label{fig:ahi_not_guided}
\end{center}
\end{figure}

\paragraph{Results and analysis} The unguided baseline, with direct access to AHI as an input, partitions the data based on its most prominent features. An analysis of the cluster profiles (see Appendix~\ref{clusters_profiles} for the full table) shows that the model excels at isolating a high-level of AHI cohort in Cluster 3, a task simplified by the direct visibility of $\boldsymbol{y}$. For the remaining population, however, the model does not find further AHI-related structure. It defaults to partitioning users based on the next largest source of variance: sleep duration. Consequently, Clusters 1 and 2 have nearly identical AHI profiles but represent different sleep behaviors. The AHI distributions in Figure~\ref{fig:ahi_not_guided} visually confirm this: while Cluster 3 is distinct, the distributions for Clusters 1 and 2 are almost perfectly overlapping.

In contrast, GCVAE uses $\boldsymbol{y}$ not as an input to be partitioned, but as a lens to find the most meaningful structure within $\boldsymbol{x}$. It is tasked with discovering groups of patients whose features in $\boldsymbol{x}$ are collectively indicative of different AHI levels. The resulting cluster profiles (see Appendix~\ref{clusters_profiles} for detailed profiles) reveal three coherent subgroups that align with a known clinical gradient of risk:

\begin{itemize}
    \item Cluster 1 (Low-Risk): A group whose profile in $\boldsymbol{x}$ (lowest age and BMI, most deep sleep) corresponds to a low AHI.
    \item Cluster 2 (Intermediate-Risk): A transitional group whose profile in $\boldsymbol{x}$ indicates a moderately increased clinical risk.
    \item Cluster 3 (High-Risk): A cohort whose profile in $\boldsymbol{x}$ (highest age and BMI, least deep sleep) corresponds to a high AHI.
\end{itemize}

\begin{figure}[ht]
\begin{center}
 \includegraphics[width=0.6\textwidth]{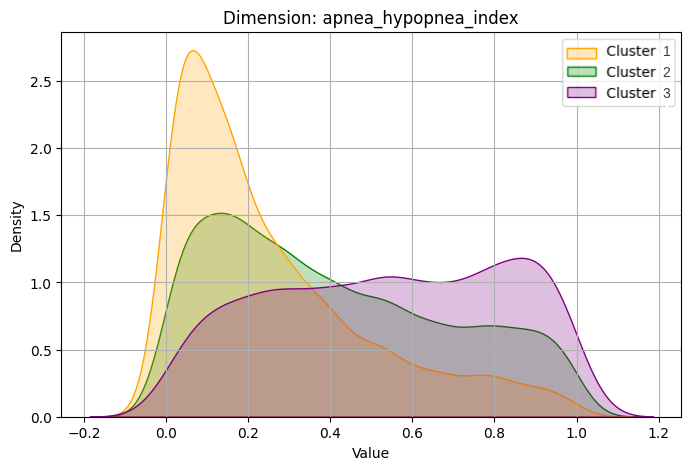}
 \caption{Distribution of the AHI within each cluster discovered by the GCVAE model.} \label{fig:ahi_guided}
\end{center}
\end{figure}

Figure~\ref{fig:ahi_guided} visualizes this result. The AHI distributions are clearly ordered, though they exhibit overlap. This overlap is a direct reflection of the challenging inference task. Unlike the baseline, which directly observes $\boldsymbol{y}$, GCVAE must infer AHI severity from complex patterns in $\boldsymbol{x}$. The overlap thus represents the inherent uncertainty of this relationship in real-world data. Crucially, this clustering is achieved solely from the sleep biometrics $\boldsymbol{x}$; unlike the baseline, the GCVAE does not have access to the guiding variable $\boldsymbol{y}$ (AHI) during inference. The key achievement is that the model successfully learned to use the guidance from $\boldsymbol{y}$ to organize $\boldsymbol{x}$ into a coherent and clinically-aligned structure.

\section{Conclusion}\label{sec:ccl}
We introduced \textbf{guided clustering}, a paradigm that formalizes the implicit guidance inherent in any practical clustering analysis. We argue that for discovery to be meaningful, the analytical context must be explicitly integrated into the model's optimization objective, rather than being relegated to an informal, trial-and-error process. The Guided Clustering Variational Autoencoder (GCVAE) is presented as one effective deep generative realization of this principle, though the paradigm itself is model-agnostic.

As an implementation of this paradigm, we proposed the GCVAE. We demonstrated that by optimizing a latent representation to be maximally informative about a guiding variable, GCVAE discovers clusters that are not only coherent but also meaningful for a given analytical task. Our experiments on both public and proprietary datasets confirmed its ability to uncover relevant structures, outperforming domain-agnostic methods and revealing clinically coherent user subgroups.

This paradigm opens several avenues for future work. The GCVAE itself can be extended by replacing its encoder and decoder with more sophisticated models, such as Graph Neural Networks or Transformers, to apply guided clustering to graph-structured or sequential data. 

By formally incorporating a guiding signal into the optimization objective, our approach provides a principled method for navigating the inherent ambiguity of what constitutes a \enquote{good} cluster, ensuring the discovered partitions are aligned with a specific analytical goal. We believe this formalization of guidance is a promising step toward developing models that can discover relevant, human-interpretable structures within vast, unannotated datasets.

{
\footnotesize
\paragraph{Data availability}
The MNIST and SVHN datasets analyzed during the current study are available in the public domain via the torchvision library. The sleep dataset used in this study is proprietary to Withings and is not publicly available due to participant privacy and GDPR restrictions; however, details are provided in the article in Appendix~\ref{appendix:dataset}.

\paragraph{Funding}
This work was supported by the Association Nationale de la Recherche et de la Technologie (ANRT) under CIFRE grant, and a private partner in the digital health sector.

\paragraph{Competing interests}
The first author is employed by a private partner. The second author declares no competing interests.

\paragraph{Author contributions}
\begin{itemize}
    \item First author: Conceptualization, methodology, software, data curation, investigation and vizualization, writing - original draft. 
    \item Second author: Conceptualization, writing - review \& editing, supervision.
\end{itemize}
}

\bibliography{biblio}
\bibliographystyle{apalike}

\clearpage
\appendix
\thispagestyle{empty}

\begin{center}

{\large\bf SUPPLEMENTARY MATERIAL}

\end{center}

\section{Development of the ELBO}\label{elbo}

\subsection{Terms of the ELBO}
Using Equation~\ref{eq:inference}, we can decompose the ELBO in five terms:
\begin{align*}
    \mathcal{L}_{\boldsymbol{\Theta}, {\boldsymbol{\phi}},\beta}(\mathbf{x},\mathbf{y}, \mathbf{z},\mathbf{c})
    &= \mathbb{E}_{q_{\boldsymbol{\phi}}(\mathbf{z}|\mathbf{x})}[ \log p_{\boldsymbol{\Theta}}(\mathbf{y}|\mathbf{z}) ] - \beta \ \text{KL}[q_{\boldsymbol{\phi}}(\mathbf{z},\mathbf{c}|\mathbf{x})||p_{\boldsymbol{\Theta}}(\mathbf{z}|\mathbf{c})p_{\boldsymbol{\Theta}}(\mathbf{c})] \\
    &= \mathbb{E}_{q_{\boldsymbol{\phi}}(\mathbf{z}|\mathbf{x})}[ \log p_{\boldsymbol{\Theta}}(\mathbf{y}|\mathbf{z}) ] 
    + \beta \left( \mathbb{E}_{q_{\boldsymbol{\phi}}(\mathbf{z},\mathbf{c}|\mathbf{x})}[ \log p_{\boldsymbol{\Theta}}(\mathbf{z}|\mathbf{c})] \right.\\
    & \hspace{1em}
    + \left. \mathbb{E}_{q_{\boldsymbol{\phi}}(\mathbf{z},\mathbf{c}|\mathbf{x})}[ \log p_{\boldsymbol{\Theta}}(\mathbf{c})] 
    - \mathbb{E}_{q_{\boldsymbol{\phi}}(\mathbf{z},\mathbf{c}|\mathbf{x})}[ \log q_{\boldsymbol{\phi}}(\mathbf{z}|\mathbf{x})] 
    - \mathbb{E}_{q_{\boldsymbol{\phi}}(\mathbf{z},\mathbf{c}|\mathbf{x})}[ \log q(\mathbf{c}|\mathbf{x})] \right) .
\end{align*}

\textbf{1st term}
\begin{align*}
\mathbb{E}_{q_{\boldsymbol{\phi}}(\mathbf{z},\mathbf{c}|\mathbf{x})}[ \log p_{\boldsymbol{\Theta}}(\mathbf{y}|\mathbf{z})] 
&= \int \sum_{\mathbf{c}} q_{\boldsymbol{\phi}}(\mathbf{z}|\mathbf{x}) q(\mathbf{c}|\mathbf{x}) \log p_{\boldsymbol{\Theta}}(\mathbf{y}|\mathbf{z}) d\mathbf{z} \\
&= \int q_{\boldsymbol{\phi}}(\mathbf{z}|\mathbf{x}) \log p_{\boldsymbol{\Theta}}(\mathbf{y}|\mathbf{z}) d\mathbf{z} \\ 
&= \mathbb{E}_{q_{\boldsymbol{\phi}}(\mathbf{z}|\mathbf{x})}[ \log p_{\boldsymbol{\Theta}}(\mathbf{y}|\mathbf{z})] \\
&= \mathbb{E}_{q_{\boldsymbol{\phi}}(\mathbf{z}|\mathbf{x})}\left[ \log \left( \frac{1}{\sqrt{(2\pi)^{d_y}}} \exp(-\frac{1}{2} ||\mathbf{y} - f_{{\boldsymbol{\theta}}}(\mathbf{z})||^2 ) \right)\right] \\
    &= -\frac{d_y}{2}\log(2\pi) -\frac{1}{2} \mathbb{E}_{q_{\boldsymbol{\phi}}(\mathbf{z}|\mathbf{x})}\left[  ||\mathbf{y} - f_{{\boldsymbol{\theta}}}(\mathbf{z})||^2 \right] .
\end{align*}

Using the SGVB estimator, we can approximate it as:
\[
\mathbb{E}_{q_{\boldsymbol{\phi}}(\mathbf{z}|\mathbf{x})}[ \log p_{\boldsymbol{\Theta}}(\mathbf{y}|\mathbf{z})] \propto - \frac{1}{L} \sum_{l=1}^L \sum_{i=1}^n ||\boldsymbol{y}_i - f_{{\boldsymbol{\theta}}}(\boldsymbol{z}_i^{(l)})||_2^2
\]

with $\boldsymbol{z}_i^{(l)} \sim \mathcal{N}(\boldsymbol{z};\tilde{\boldsymbol{\mu}}_i, \tilde{\boldsymbol{\sigma}}_i^2I)$, and $[\tilde{\boldsymbol{\mu}}_i, \log \tilde{\boldsymbol{\sigma}}_i^2] = g_{{\boldsymbol{\phi}}}(\boldsymbol{x}_i)$. $L$ is the number of Monte Carlo samples in the SGVB estimator.

\textbf{2nd term}
\begin{align*}
\mathbb{E}_{q_{\boldsymbol{\phi}}(\mathbf{z},\mathbf{c}|\mathbf{x})}[ \log p_{\boldsymbol{\Theta}}(\mathbf{z}|\mathbf{c})] 
&= \int \sum_{\mathbf{c}} q_{\boldsymbol{\phi}}(\mathbf{z}|\mathbf{x}) q(\mathbf{c}|\mathbf{x}) \log p_{\boldsymbol{\Theta}}(\mathbf{z}|\mathbf{c}) d\mathbf{z} \\ 
&= \sum_{\mathbf{c}} q(\mathbf{c}|\mathbf{x}) \int \mathcal{N}(\mathbf{z}; \tilde{\boldsymbol{\mu}}, \tilde{\boldsymbol{\sigma}}^2I) \log \mathcal{N}(\mathbf{z}; \boldsymbol{\mu}_c, \boldsymbol{\sigma}_c^2I) d\mathbf{z} .
\end{align*}
Using Lemma~\ref{lem:usefullemma} in Appendix~\ref{appendix:lemma}, we have:
\begin{align*}
\mathbb{E}_{q_{\boldsymbol{\phi}}(\mathbf{z},\mathbf{c}|\mathbf{x})}[ \log p_{\boldsymbol{\Theta}}(\mathbf{z}|\mathbf{c})] 
&= - \sum_{\mathbf{c}} q(\mathbf{c}|\mathbf{x}) \left[\frac{J}{2} \log(2\pi) + \frac{1}{2} \sum_{j=1}^J \left(\log \sigma_{cj}^2-\frac{\tilde{\sigma}_j^2}{\sigma_{cj}^2}-\frac{\left(\tilde{\mu}_j-\mu_{cj}\right)^2}{\sigma_{cj}^2} \right) \right]\\ 
&= -\frac{J}{2} \log(2\pi) - \frac{1}{2} \sum_{\mathbf{c}} q(\mathbf{c}|\mathbf{x}) \sum_{j=1}^J \left(\log \sigma_{cj}^2 + \frac{\tilde{\sigma}_j^2}{\sigma_{cj}^2} + \frac{\left(\tilde{\mu}_j-\mu_{cj}\right)^2}{\sigma_{cj}^2} \right) .
\end{align*}

\textbf{3rd term}
\begin{align*}
\mathbb{E}_{q_{\boldsymbol{\phi}}(\mathbf{z},\mathbf{c}|\mathbf{x})}[ \log p_{\boldsymbol{\Theta}}(\mathbf{c})] 
&= \int \sum_{\mathbf{c}} q_{\boldsymbol{\phi}}(\mathbf{z}|\mathbf{x}) q(\mathbf{c}|\mathbf{x}) \log p_{\boldsymbol{\Theta}}(\mathbf{c}) d\mathbf{z} \\
&= \sum_{\mathbf{c}} q(\mathbf{c}|\mathbf{x}) \log \pi_{\mathbf{c}}.
\end{align*}

\textbf{4th term}
\begin{align*}
\mathbb{E}_{q_{\boldsymbol{\phi}}(\mathbf{z},\mathbf{c}|\mathbf{x})}[ \log q_{\boldsymbol{\phi}}(\mathbf{z}|\mathbf{x})] 
&= \int \sum_{\mathbf{c}} q_{\boldsymbol{\phi}}(\mathbf{z}|\mathbf{x}) q(\mathbf{c}|\mathbf{x}) \log q_{\boldsymbol{\phi}}(\mathbf{z}|\mathbf{x}) d\mathbf{z} \\
&= \int \mathcal{N}(\mathbf{z}; \tilde{\boldsymbol{\mu}}, \tilde{\boldsymbol{\sigma}}^2I) \log \mathcal{N}(\mathbf{z}; \tilde{\boldsymbol{\mu}}, \tilde{\boldsymbol{\sigma}}^2I) d\mathbf{z} .
\end{align*}
Using Lemma \ref{lem:usefullemma} in Appendix \ref{appendix:lemma}, we have:
\[
\mathbb{E}_{q_{\boldsymbol{\phi}}(\mathbf{z},\mathbf{c}|\mathbf{x})}[ \log q_{\boldsymbol{\phi}}(\mathbf{z}|\mathbf{x})] 
= - \frac{J}{2} \log(2\pi) - \frac{1}{2} \sum_{j=1}^J(1+\log \tilde{\sigma}_j^2) .
\]

\textbf{5th term}
\begin{align*}
\mathbb{E}_{q_{\boldsymbol{\phi}}(\mathbf{z},\mathbf{c}|\mathbf{x})}[ \log q(\mathbf{c}|\mathbf{x})] 
&= \int \sum_{\mathbf{c}} q_{\boldsymbol{\phi}}(\mathbf{z}|\mathbf{x}) q(\mathbf{c}|\mathbf{x}) \log q(\mathbf{c}|\mathbf{x}) d\mathbf{z} \\
&= \sum_{\mathbf{c}} q(\mathbf{c}|\mathbf{x}) \log q(\mathbf{c}|\mathbf{x}).
\end{align*}

\subsection{Lemma} \label{appendix:lemma}

As presented in \cite{VADE}, we have Lemma \ref{lem:usefullemma}. 
\begin{lemma}
\label{lem:usefullemma}
Given two multivariate Gaussian distributions $q(z)=\mathcal{N}\left(z ; \tilde{\mu}, \tilde{\sigma}^2 I\right)$ and $p(z)=\mathcal{N}\left(z ; \mu, \sigma^2 I\right)$, we have:
\[
\int q(z) \log p(z) d z
=-\frac{J}{2} \log(2\pi) - \frac{1}{2} \sum_{j=1}^J \left(\log \sigma_j^2 + \frac{\tilde{\sigma}_j^2}{ \sigma_j^2} + \frac{\left(\tilde{\mu}_j-\mu_j\right)^2}{ \sigma_j^2} \right)
\]
where $\mu_j, \sigma_j, \tilde{\mu}_j$ and $\tilde{\sigma}_j$ simply denote the $j^{\text {th }}$ element of $\mu, \sigma, \tilde{\mu}$ and $\tilde{\sigma}$, respectively, and $J$ is the dimensionality of $z$.
\end{lemma}

\begin{proof} 
{\allowdisplaybreaks
\begin{align*}
& \int q(z) \log p(z) d z = \int \mathcal{N}\left(z ; \tilde{\mu}, \tilde{\sigma}^2 I\right) \log \mathcal{N}\left(z ; \mu, \sigma^2 I\right) d z \\
& =\int \prod_{j=1}^J \frac{1}{\sqrt{2 \pi \tilde{\sigma}_j^2}} \exp \left(-\frac{\left(z_j-\tilde{\mu}_j\right)^2}{2 \tilde{\sigma}_j^2}\right) \log \left[\prod_{j=1}^J \frac{1}{\sqrt{2 \pi \sigma_j^2}} \exp \left(-\frac{\left(z_j-\mu_j\right)^2}{2 \sigma_j^2}\right)\right] d z \\
& =\sum_{j=1}^J \int \frac{1}{\sqrt{2 \pi \tilde{\sigma}_j^2}} \exp \left(-\frac{\left(z_j-\tilde{\mu}_j\right)^2}{2 \tilde{\sigma}_j^2}\right) \log \left[\frac{1}{\sqrt{2 \pi \sigma_j^2}} \exp \left(-\frac{\left(z_j-\mu_j\right)^2}{2 \sigma_j^2}\right)\right] d z_j \\
& =\sum_{j=1}^J \int \frac{1}{\sqrt{2 \pi \tilde{\sigma}_j^2}} \exp \left(-\frac{\left(z_j-\tilde{\mu}_j\right)^2}{2 \tilde{\sigma}_j^2}\right)\left[-\frac{1}{2} \log \left(2 \pi \sigma_j^2\right)\right] d z_j \\
&\hspace{1cm} -\int \frac{1}{\sqrt{2 \pi \tilde{\sigma}_j^2}} \exp \left(-\frac{\left(z_j-\tilde{\mu}_j\right)^2}{2 \tilde{\sigma}_j^2}\right) \frac{\left(z_j-\mu_j\right)^2}{2 \sigma_j^2} d z_j \\
& =\sum_{j=1}^J-\frac{1}{2} \log \left(2 \pi \sigma_j^2\right)\\
&\hspace{1cm} -\int \frac{1}{\sqrt{2 \pi \tilde{\sigma}_j^2}} \exp \left(-\frac{\left(z_j-\tilde{\mu}_j\right)^2}{2 \tilde{\sigma}_j^2}\right) \frac{\left(z_j-\tilde{\mu}_j\right)^2+2\left(z_j-\tilde{\mu}_j\right)\left(\tilde{\mu}_j-\mu_j\right)+\left(\tilde{\mu}_j-\mu_j\right)^2}{2 \tilde{\sigma}_j^2} \frac{\tilde{\sigma}_j^2}{\sigma_j^2} d z_j \\
& = C
-\frac{\tilde{\sigma}_j^2}{\sigma_j^2} \int \frac{1}{\sqrt{2 \pi \tilde{\sigma}_j^2}} \exp \left(-\frac{\left(z_j-\tilde{\mu}_j\right)^2}{2 \tilde{\sigma}_j^2}\right) \frac{\left(z_j-\tilde{\mu}_j\right)^2}{2 \tilde{\sigma}_j^2} d z_j \\
&\hspace{1cm}- \frac{\left(\tilde{\mu}_j-\mu_j\right)}{2 \sigma_j^2} \underbrace{\int \frac{1}{\sqrt{2 \pi \tilde{\sigma}_j^2}} \exp \left(-\frac{\left(z_j-\tilde{\mu}_j\right)^2}{2 \tilde{\sigma}_j^2}\right) \left(z_j-\tilde{\mu}_j\right)  d z_j}_\text{$=0$} \\
&\hspace{1cm} - \frac{\left(\tilde{\mu}_j-\mu_j\right)^2}{2 \sigma_j^2}  \underbrace{\int \frac{1}{\sqrt{2 \pi \tilde{\sigma}_j^2}} \exp \left(-\frac{\left(z_j-\tilde{\mu}_j\right)^2}{2 \tilde{\sigma}_j^2}\right)  d z_j}_\text{$=1$} \\
& = C-\frac{\tilde{\sigma}_j^2}{\sigma_j^2} \int \frac{1}{\sqrt{2 \pi \tilde{\sigma}_j^2}} \exp \left(-\frac{\left(z_j-\tilde{\mu}_j\right)^2}{2 \tilde{\sigma}_j^2}\right) \frac{\left(z_j-\tilde{\mu}_j\right)^2}{2 \tilde{\sigma}_j^2} d z_j-\frac{\left(\tilde{\mu}_j-\mu_j\right)^2}{2 \sigma_j^2}  \\
& =C-\frac{\tilde{\sigma}_j^2}{\sigma_j^2} \int \frac{1}{\sqrt{2 \pi}} \exp \left(-\frac{x_j^2}{2}\right) \frac{x_j^2}{2} d x_j-\frac{\left(\tilde{\mu}_j-\mu_j\right)^2}{2 \sigma_j^2} \tag*{\textit{By change of variables}} \\
& =C-\frac{\tilde{\sigma}_j^2}{\sigma_j^2} \int \frac{1}{\sqrt{2 \pi}}\left(-\frac{x_j}{2}\right) d\left(\exp \left(-\frac{x_j^2}{2}\right)\right)-\frac{\left(\tilde{\mu}_j-\mu_j\right)^2}{2 \sigma_j^2} \tag*{\textit{By change of variables}}\\
& = C - \frac{\tilde{\sigma}_j^2}{\sigma_j^2} \left[ \left[\frac{1}{\sqrt{2 \pi}}\left(-\frac{x_j}{2}\right) \exp \left(-\frac{x_j^2}{2}\right)\right]_{-\infty} ^{+\infty} - \int \frac{1}{\sqrt{2 \pi}} \exp \left(-\frac{x_j^2}{2}\right) d\left(-\frac{x_j}{2}\right)\right]-\frac{\left(\tilde{\mu}_j-\mu_j\right)^2}{2 \sigma_j^2} \tag*{\textit{By integration by parts}}\\
& = C - \frac{\tilde{\sigma}_j^2}{\sigma_j^2} \left[ (0-0) - (-\frac{1}{2}) \underbrace{\int \frac{1}{\sqrt{2 \pi}} \exp \left(-\frac{x_j^2}{2}\right) dx_j}_{=1} \right]-\frac{\left(\tilde{\mu}_j-\mu_j\right)^2}{2 \sigma_j^2}\\
& =\sum_{j=1}^J-\frac{1}{2} \log \left(2 \pi \sigma_j^2\right)-\frac{\tilde{\sigma}_j^2}{2 \sigma_j^2}-\frac{\left(\tilde{\mu}_j-\mu_j\right)^2}{2 \sigma_j^2} \\
&= -\frac{J}{2} \log(2\pi) - \frac{1}{2} \sum_{j=1}^J \left(\log \sigma_j^2 + \frac{\tilde{\sigma}_j^2}{\sigma_j^2}+\frac{\left(\tilde{\mu}_j-\mu_j\right)^2}{\sigma_j^2} \right)
\end{align*}}
where $C$ denotes $\sum_{j=1}^J-\frac{1}{2} \log \left(2 \pi \sigma_j^2\right)$ for simplicity.
\end{proof}

\section{Approximate $q(\mathbf{c}|\mathbf{x})$}\label{approx}

We describe how to formulate $q(\mathbf{c}|\mathbf{x})$ to maximize the ELBO. Specifically, our objective can be rewritten as:

\begin{align*}
\mathcal{L}_{\boldsymbol{\Theta}, {\boldsymbol{\phi}},\beta}(\mathbf{x},\mathbf{y}, \mathbf{z},\mathbf{c})
&= \mathbb{E}_{q_{\boldsymbol{\phi}}(\mathbf{z},\mathbf{c}|\mathbf{x})} [\log p_{\boldsymbol{\theta}}(\mathbf{y}|\mathbf{z},\mathbf{c})] 
- \beta \ \text{KL}[q_{\boldsymbol{\phi}}(\mathbf{z},\mathbf{c}|\mathbf{x})||p_{\boldsymbol{\Theta}}(\mathbf{z},\mathbf{c})] \\
&= \int \sum_{\mathbf{c}} q_{\boldsymbol{\phi}}(\mathbf{z}|\mathbf{x})q(\mathbf{c}|\mathbf{x}) \log \frac{p_{\boldsymbol{\Theta}}(\mathbf{y}|\mathbf{z}) p_{\boldsymbol{\Theta}}(\mathbf{c}|\mathbf{z})^\beta p_{\boldsymbol{\Theta}}(\mathbf{z})^\beta}{q_{\boldsymbol{\phi}}(\mathbf{z}|\mathbf{x})^\beta q(\mathbf{c}|\mathbf{x})^\beta} d\mathbf{z} \\
&= \int q_{\boldsymbol{\phi}}(\mathbf{z}|\mathbf{x}) \log \frac{p_{\boldsymbol{\Theta}}(\mathbf{y}|\mathbf{z})p_{\boldsymbol{\Theta}}(\mathbf{z})^\beta}{q_{\boldsymbol{\phi}}(\mathbf{z}|\mathbf{x})^\beta} d\mathbf{z} 
+ \int q_{\boldsymbol{\phi}}(\mathbf{z}|\mathbf{x}) \sum_{\mathbf{c}} q(\mathbf{c}|\mathbf{x}) \log \frac{p_{\boldsymbol{\Theta}}(\mathbf{c}|\mathbf{z})^\beta}{q(\mathbf{c}|\mathbf{x})^\beta} d\mathbf{z} \\
&= \int q_{\boldsymbol{\phi}}(\mathbf{z}|\mathbf{x}) \log \frac{p_{\boldsymbol{\Theta}}(\mathbf{y}|\mathbf{z})p_{\boldsymbol{\Theta}}(\mathbf{z})^\beta}{q_{\boldsymbol{\phi}}(\mathbf{z}|\mathbf{x})^\beta} d\mathbf{z} 
- \int q_{\boldsymbol{\phi}}(\mathbf{z}|\mathbf{x}) \ \beta \ \text{KL}[q(\mathbf{c}|\mathbf{x})||p_{\boldsymbol{\Theta}}(\mathbf{c}|\mathbf{z})] d\mathbf{z} .
\end{align*}

As in \cite{VADE}, the first term does not depend on $\mathbf{c}$ and the second term is non-negative. Thus, maximizing the lower bound ELBO with respect to $q(\mathbf{c}|\mathbf{x})$ requires that $\text{KL}[q(\mathbf{c}|\mathbf{x})||p_{\boldsymbol{\Theta}}(\mathbf{c}|\mathbf{z})] = 0$. Thus, with $\nu$ a constant, we have:
\[
\frac{q(\mathbf{c}|\mathbf{x})}{p_{\boldsymbol{\Theta}}(\mathbf{c}|\mathbf{z})} = \nu .
\]

Since $\sum_{\mathbf{c}} q(\mathbf{c}|\mathbf{x}) = 1$ and $\sum_{\mathbf{c}} p_{\boldsymbol{\Theta}}(\mathbf{c}|\mathbf{z}) = 1$, we have:
\[
\frac{q(\mathbf{c}|\mathbf{x})}{p_{\boldsymbol{\Theta}}(\mathbf{c}|\mathbf{z})} = 1 .
\]
Taking the expectation on both sides, we can obtain:
\[
q(\mathbf{c}|\mathbf{x}) = \mathbb{E}_{q_{\boldsymbol{\phi}}(\mathbf{z}|\mathbf{x})} [p_{\boldsymbol{\Theta}}(\mathbf{c}|\mathbf{z})] .
\]

\section{Verification of GVI assumptions}\label{appendix:cv}

In this section, we provide the detailed verification that the GCVAE model satisfies the assumptions required for Theorem~2 (Consistency under Independence) in \cite{knoblauch2019frequentist}.

We denote the global parameters by $\boldsymbol{\psi} = \{\boldsymbol{\Theta}, {\boldsymbol{\phi}}\} \in \boldsymbol{\Psi}$, where $\boldsymbol{\Theta}$ contains the generative parameters (GMM and decoder weights) and ${\boldsymbol{\phi}}$ contains the variational parameters (encoder weights). 
We denote by $\mathcal{X} = \mathcal{X}^o \times \mathcal{Z}$, where $\mathcal{X}^o$ and $\mathcal{Z}$ denote the spaces of the observables $(\boldsymbol{x}_i, \boldsymbol{y}_i)$ and the latent components $(\boldsymbol{z}_{i}, \boldsymbol{c}_i)$.

\paragraph{Assumption 1.} The GVI problem is well-defined.

\begin{enumerate}
    \item \textit{The loss function $\mathcal{L}: \boldsymbol{\Psi} \times \mathcal{X} \rightarrow \mathbb{R}$ is discontinuous at most at finitely many points.} \\
    By using the Gumbel-Softmax relaxation, the loss function is a composition of continuous functions and is therefore continuous everywhere (0 points of discontinuity). 
    \item \textit{For any $(\boldsymbol{x}_i, \boldsymbol{y}_i) \in \mathcal{X}^o$ and any $n, \mathcal{L}\left(\boldsymbol{\psi}, \boldsymbol{x}_i, \boldsymbol{y}_i, \boldsymbol{z}_{i}, \boldsymbol{c}_i\right)<\infty$ for all $(\boldsymbol{z}_{i}, \boldsymbol{c}_i) \in \mathcal{Z}$.} \\
    Since the log-densities of the Gaussian and Concrete distributions are finite everywhere on their respective open supports (assuming $\boldsymbol{\sigma} > 0$, $\boldsymbol{\pi} > 0$ and $\boldsymbol{c}_k > 0$), the pointwise loss value is finite for any generated sample.
    \item \textit{The minimizers $\hat{\boldsymbol{\psi}}_n=\arg \min_{\boldsymbol{\psi}}\left\{\frac{1}{n} \sum_{i=1}^n \mathcal{L}\left(\boldsymbol{\psi}, \boldsymbol{x}_i, \boldsymbol{y}_i, \boldsymbol{z}_{i}, \boldsymbol{c}_i\right)\right\} \in \boldsymbol{\Psi}$ exist for all $n$.} \\
    Since the loss function is continuous and defined on a compact set $\boldsymbol{\Psi}$ (by hypothesis), the Weierstrass extreme value theorem guarantees that a global minimizer exists.
\end{enumerate}
For a probability measure $\mu$ on $\mathcal{X}$:
\begin{enumerate}
  \setcounter{enumi}{3}
    \item \textit{The loss satisfies a law of large numbers, i.e $\frac{1}{n} \sum_{i=1}^n \mathcal{L}\left(\boldsymbol{\psi}, \boldsymbol{x}_i, \boldsymbol{y}_i, \boldsymbol{z}_{i}, \boldsymbol{c}_i\right) \xrightarrow{\mu-\text { a.s. }} \mathbb{E}_\mu[\mathcal{L}(\boldsymbol{\psi}, \boldsymbol{x}, \boldsymbol{y}, \boldsymbol{z}, \boldsymbol{c})]$.} \\
    We make the standard statistical learning assumption that the training observations $\{(\boldsymbol{x}_i, \boldsymbol{y}_i)\}_{i=1}^n$ are i.i.d. Under the i.i.d. assumption, the law of large numbers applies.
    \item \textit{The $\mu$-population-minimizer $\boldsymbol{\psi}^*=\arg \min_{\boldsymbol{\psi}} \mathbb{E}_\mu[\mathcal{L}(\boldsymbol{\psi}, \boldsymbol{x}, \boldsymbol{y}, \boldsymbol{z}, \boldsymbol{c})] \in \boldsymbol{\Psi}$ exists and is unique.} 
    \begin{itemize}
        \item Existence: As established in Assumption 1.1, the pointwise loss function is continuous. Under the integrability condition (Assumption 1.2), this implies the population risk is continuous with respect to $\boldsymbol{\psi}$. Since the parameter space $\boldsymbol{\Psi}$ is compact, the Weierstrass extreme value theorem guarantees that it attains its minimum. Thus, a minimizer exists.
        \item Uniqueness: Neural networks are inherently non-identifiable due to permutation symmetries.However, we make the standard identifiability assumption that the global minimum is unique up to these permutation symmetries. That is, we assume that any two distinct parameter vectors achieving the minimum are functionally equivalent and differ only by a permutation of their indices.
    \end{itemize}
    \item \textit{The loss is finite in $\mu$-expectation, i.e. $\mathbb{E}_\mu[\mathcal{L}(\boldsymbol{\psi}, \boldsymbol{x}, \boldsymbol{y}, \boldsymbol{z}, \boldsymbol{c})]<\infty$ for all $\boldsymbol{\psi} \in \boldsymbol{\Psi}$.} \\
    We previously established that the loss is finite pointwise. The integration over the latent space $\mathcal{Z}$ is well-defined because the Gaussian variational distributions have finite moments. Furthermore, since the neural networks are Lipschitz continuous on the compact parameter space $\boldsymbol{\Psi}$, the loss grows at most quadratically with respect to the observations $(\boldsymbol{x}, \boldsymbol{y})$. Therefore, under the standard assumption that the data distribution has finite second moments, the total expectation is finite.
    \item \textit{One of the following holds true: $\mathbb{E}_\mu[\mathcal{L}(\boldsymbol{\psi}, \boldsymbol{x}, \boldsymbol{y}, \boldsymbol{z}, \boldsymbol{c})]$ is coercive in $\boldsymbol{\psi}$ or $\boldsymbol{\Psi}$ is compact.} \\
    We define $\boldsymbol{\Psi}$ as a compact set.
\end{enumerate}

\paragraph{Assumption 2} \textit{The variational family $\mathcal{Q}=\mathcal{Q}^{\boldsymbol{\psi}} \times \mathcal{Q}^z$ with $\mathcal{Q}^{\boldsymbol{\psi}}=\{q(\boldsymbol{\psi} \mid \boldsymbol{\kappa}): \boldsymbol{\kappa} \in \boldsymbol{K}\}$ and $\mathcal{Q}^z= \left\{q^z\left(\boldsymbol{z}, \boldsymbol{c} \mid \boldsymbol{\eta}\right): \boldsymbol{\eta} \in \boldsymbol{H}\right\}$ consists of absolutely continuous densities with respect to the Lebesgue measure. Moreover, for all $n \in \mathbb{N}$ and any $\left(\boldsymbol{\psi}^*, \boldsymbol{z}, \boldsymbol{c}\right) \in \boldsymbol{\Psi} \times \mathcal{Z}$, there exist sequences $\left\{\boldsymbol{\kappa}_k\right\}_{k=1}^{\infty}$ and $\left\{\boldsymbol{\eta}_k\right\}_{k=1}^{\infty}$ of variational parameters so that $q\left(\boldsymbol{\psi} \mid \boldsymbol{\kappa}_k\right) \xrightarrow{\mathcal{D}}  \delta_{\boldsymbol{\psi}^*}(\boldsymbol{\psi})$ and $q^z\left(\boldsymbol{z}, \boldsymbol{c} \mid \boldsymbol{\eta}_k\right) \xrightarrow{\mathcal{D}} \delta_z\left(\boldsymbol{z}, \boldsymbol{c} \right)$ as $k \rightarrow \infty$.}

The proposed variational family satisfies these regularity conditions by design. First, regarding absolute continuity, the employment of the Gumbel-Softmax relaxation endows the latent distribution with a valid, absolutely continuous density with respect to the Lebesgue measure on the simplex. The weak convergence to the Dirac measure is achieved in the distributional limit as the temperature $\tau \to 0$ and variance $\boldsymbol{\sigma}^2 \to 0$.

\paragraph{Assumption 3} \textit{The GVI uncertainty quantifier $D: \mathcal{P}(\boldsymbol{\Psi})^2 \rightarrow \mathbb{R}_{+}$ is a statistical divergence. Further, it is lower semi-continuous in its first argument with respect to the weak topology of $\mathcal{P}(\boldsymbol{\Psi})$.}

We employ the Kullback-Leibler (KL) divergence as the uncertainty quantifier. The KL divergence is a well-defined statistical divergence and is rigorously known to be lower semi-continuous. Thus, the assumption is satisfied.

\paragraph{Assumption 4} \textit{The prior $\pi$ and the GVI uncertainty quantifier $D$ are suitable for the variational family $\mathcal{Q}^{\boldsymbol{\psi}}$ : For all $q \in \mathcal{Q}^{\boldsymbol{\psi}}, D(q \| \pi)<\infty$.}

The prior is modeled as a Gaussian Mixture Model, and the variational family $\mathcal{Q}$ is a mean-field product of Gaussian and Gumbel-Softmax distributions. Since the variational components share the same support as the prior and possess finite moments, the KL divergence is well-defined and finite for all valid variational parameters.

\paragraph{Assumption 5} \textit{The prior belief $\pi_prior$ about $\boldsymbol{\psi}$ is not infinitely bad: $\mathbb{E}_{\pi_{prior}}\left[\mathbb{E}_\mu[\mathcal{L}(\boldsymbol{\psi}, \boldsymbol{x}, \boldsymbol{y}, \boldsymbol{z}, \boldsymbol{c})] \right] = C_{\pi_{prior}}<\infty$. Moreover, $\mathcal{Q}^{\boldsymbol{\psi}}$ contains the singleton ${\pi_{prior}}(\boldsymbol{\psi})$. In other words, $\mathcal{Q}^{\boldsymbol{\psi}}=\{q(\boldsymbol{\psi} \mid \boldsymbol{\kappa}): \boldsymbol{\kappa} \in \boldsymbol{K}\} \cup\{{\pi_{prior}}(\boldsymbol{\psi})\}$.} \\
The prior ${\pi_{prior}}(\boldsymbol{\psi})$ is modeled as a standard Gaussian. Since the loss function is continuous and grows polynomially with respect to the weights (Lipschitz neural networks), and the Gaussian prior has finite moments of all orders, the expected loss under the prior is strictly finite.
The second condition is satisfied by construction.

\paragraph{Assumption 7} \textit{There exists a compact subset $A \subset \boldsymbol{\Psi}$ so that (i) $\boldsymbol{\psi}^* \in A$ and (ii) $\pi_{prior} \geq \bar{q}_n$ on $\boldsymbol{\Psi} \backslash A$, for all $n \geq N$ for some $N<\infty$.}  \\
Since we explicitly defined the parameter space $\boldsymbol{\Psi}$ to be compact (Assumption 1.7), we can choose the subset $A$ to be the entire space, i.e., $A = \boldsymbol{\Psi}$.
\begin{enumerate}[(i)]
\item The minimizer $\boldsymbol{\psi}^*$ must reside within the valid parameter space $\boldsymbol{\Psi}$.
\item The complement set $\boldsymbol{\Psi} \backslash A$ is the empty set.
\end{enumerate}

Therefore our model verifies the assumption of the theorem: 
\begin{theorem}[GVI consistency under independence]
    If Assumptions 1, 2, 3, 4, 5 and 7 hold and $(\boldsymbol{x}_i, \boldsymbol{y}_i) \stackrel{i i d}{\sim} (\boldsymbol{x}_1, \boldsymbol{y}_1)$, then the GVI posteriors are consistent. i.e., $q_n \xrightarrow{\mathcal{D}} \delta_{\boldsymbol{\psi}^*}$  $\mu$-almost surely, where $\mu$ is the probability measure on $(\boldsymbol{x}_1, \boldsymbol{y}_1)$.
\end{theorem}

\section{Withings dataset}\label{appendix:dataset}

The dataset is composed of 50,000 individuals, each contributing one night of sleep data to avoid bias from repeated measurements for the same user. 
The data was recorded by the Withings Sleep Analyzer\footnote{\url{https://www.withings.com/us/en/sleep}}, capturing comprehensive sleep and biometric information. The dataset includes an equal number of users across the three categories based on the Apnea-Hypopnea Index (AHI) as categorized by the American Academy of Sleep Medicine (AASM) \cite{ahi}: no or mild ($<$15 events/hour), moderate (15–30 events/hour), and severe ($>$30 events/hour).

\subsection{Variables overview}
The seven variables in the input vector $\boldsymbol{x}$ are:
\begin{itemize}
    \item \textit{sleep\_duration}: Duration of the user's sleep during the night.
    \item \textit{light\_sleep\_duration}: Duration of the user's light sleep during the night.
    \item \textit{deep\_sleep\_duration}: Duration of the user's deep sleep during the night.
    \item \textit{nb\_sleep\_interruptions}: Count of awakenings throughout the night.
    \item \textit{avg\_night\_hr}: Mean heart rate during the night.
    \item \textit{bmi}: User's average Body Mass Index measured over a year. 
    \item \textit{age}: User’s age.
\end{itemize}

The guiding variable $y$ is the \textit{apnea\_hypopnea\_index} (AHI), categorizing sleep apnea severity based on the number of apnea-hypopnea events per hour.

\subsection{Dataset Summary}
In summary, Table \ref{tab:dataset_summary} shows a concise view of each variable’s range, mean, and standard deviation. 

\begin{table}[!ht]
\centering
\caption{Descriptive statistics of the variables}
\label{tab:dataset_summary}
\begin{tabular}{l|ccc}
\hline
\textbf{Variable} (unit) & \textbf{Range} & \textbf{Mean} & \textbf{Std. Dev.} \\
\hline
\textit{sleep\_duration}   (seconds)      & 14880 -- 36000 & 26224 & 4224 \\
\textit{light\_sleep\_duration}  (seconds) & 3600 -- 31860  & 15747 & 4651 \\
\textit{deep\_sleep\_duration}  (seconds)  & 3600 -- 32220  & 10472 & 4022 \\
\textit{nb\_sleep\_interruptions}          & 0 -- 20        & 2.74  & 2.34 \\
\textit{avg\_night\_hr}    (bpm)           & 40 -- 111      & 62.49 & 8.57 \\
\textit{bmi}                (kg/m²)      & 16 -- 50       & 27.53 & 5.14 \\
\textit{age}              (years)        & 18 -- 80       & 50    & 12.67 \\
\textit{apnea\_hypopnea\_index}            & 0 -- 40        & 18.14 & 13.47 \\
\hline
\end{tabular}
\end{table}

We applied Min-Max normalization to all the variables to ensure they fall on a comparable scale.

\section{Clusters profile tables for the Withings case study}\label{clusters_profiles}

This section provides the detailed cluster profiles for the user subgroups discovered by the unguided baseline and our proposed Guided Model (GCVAE), as discussed in Section \ref{subsec:sleep_dataset}. All values for user features are denormalized and presented as Mean ± Standard Deviation. 

\begin{table}[ht!]
\centering
\caption{Cluster profiles for the \textbf{unguided model}.}
\label{tab:unguided_profiles}
\begin{tabular}{l|ccc}
\toprule
 & \textbf{Cluster 1} & \textbf{Cluster 2} & \textbf{Cluster 3} \\
\midrule
AHI & 9.43 $\pm$ 6.92 & 28.55 $\pm$ 7.50 & 8.55 $\pm$ 6.69 \\
Age (years) & 49.61 $\pm$ 11.52 & 55.91 $\pm$ 11.78 & 49.36 $\pm$ 12.48 \\
BMI (kg/m²) & 27.87 $\pm$ 5.08 & 29.25 $\pm$ 5.16 & 26.93 $\pm$ 4.84 \\
Sleep Duration (hrs) & 6.63 $\pm$ 0.81 & 7.30 $\pm$ 1.18 & 8.28 $\pm$ 0.83 \\
Deep Sleep (hrs) & 2.94 $\pm$ 1.01 & 2.91 $\pm$ 1.16 & 3.32 $\pm$ 1.15 \\
Light Sleep (hrs) & 3.69 $\pm$ 1.00 & 4.39 $\pm$ 1.27 & 4.96 $\pm$ 1.19 \\
Sleep Interruptions & 2.28 $\pm$ 2.09 & 3.04 $\pm$ 2.50 & 2.75 $\pm$ 2.25 \\
Avg. Night HR (bpm) & 63.66 $\pm$ 8.59 & 63.99 $\pm$ 8.50 & 61.98 $\pm$ 8.38 \\
Avg. Resp. Rate (rpm) & 15.11 $\pm$ 2.14 & 15.11 $\pm$ 2.26 & 15.06 $\pm$ 2.10 \\
\bottomrule
\end{tabular}
\end{table}

\begin{table}[ht!]
\centering
\caption{Clinical profiles of user subgroups discovered by \textbf{GCVAE}.}
\label{tab:guided_profiles}
\begin{tabular}{l|ccc}
\toprule
 & \textbf{Cluster 1}   & \textbf{Cluster 2} & \textbf{Cluster 3} \\
\midrule
AHI & 10.4 ± 9.6 & 16.4 ± 10.8 & 21.6 ± 11.2 \\
Age (years) & 42.58 ± 9.88   & 53.35 ± 10.39 & 61.90 ± 9.73 \\
BMI (kg/m²)& 25.35 ± 3.67    & 28.25 ± 4.57  & 31.79 ± 5.78 \\

Sleep Duration (hrs) & 7.52 ± 1.12 & 7.36 ± 1.18 & 7.19 ± 1.19 \\
Deep Sleep (hrs)     & 3.41 ± 1.11 & 3.04 ± 1.09 & 2.50 ± 0.98 \\
Light Sleep (hrs)    & 4.10 ± 1.23 & 4.32 ± 1.27 & 4.68 ± 1.25 \\
Sleep Interruptions  & 2.03 ± 1.89 & 2.64 ± 2.17 & 3.83 ± 2.74 \\

Avg. Night HR (bpm)   & 61.49 ± 8.08 & 63.66 ± 8.54 & 65.01 ± 8.77 \\
Avg. Resp. Rate (rpm) & 15.08 ± 2.05 & 15.09 ± 2.19 & 15.15 ± 2.29 \\
\bottomrule
\end{tabular}
\end{table}

\end{document}